%% file: main.tex
\documentclass[tbtags,authorcolumns,numberwithinsect,footnotebackref,]{no-lipics}

\usepackage{amssymb}
\usepackage{mathtools}
\usepackage{bm}
\usepackage{stmaryrd}
\usepackage[draft]{fixme}
\usepackage{booktabs}
\usepackage[T1]{fontenc}
\usepackage[utf8]{inputenc}
\usetikzlibrary{calc}
\usepackage{lineno}

\input{macros}

\title{Hardness of Dynamic Tree Edit Distance and Friends}

\author{Bingbing Hu}{University of California, San Diego\\La Jolla, United States}{bingbing@ucsd.edu}{https://orcid.org/0000-0001-8854-5969qw}{}
\author{Jakob Nogler}{Massachusetts Institute of Technology\\Cambridge, United States}{jnogler@mit.edu}{https://orcid.org/0009-0002-7028-2595}{}
\author{Barna Saha}{University of California, San Diego\\La Jolla, United States}{barnas@ucsd.edu}{https://orcid.org/0000-0002-6494-3839}{}

\authorrunning{B. Hu, J. Nogler, and B. Saha}

\nolinenumbers

\begin{document}
% \linenumbers
\pagenumbering{roman}
\maketitle
\begin{abstract}
\input{abstract}
\end{abstract}

% \clearpage
% \thispagestyle{plain}
% \tableofcontents
% \clearpage
\pagenumbering{arabic}

\input{s1_intro.tex}
\input{s2_overview.tex}
\input{s3_prelims.tex}
\input{s4_unweighted}
\input{s5_weighted}
\input{s6_rnafolding}

\bibliographystyle{alphaurl}
\bibliography{main}

\newpage

\appendix
\input{appendix}

\end{document}

%% file: macros.tex
\usetikzlibrary{arrows,arrows.meta,shapes.misc, decorations.pathreplacing,decorations.pathmorphing,calligraphy, patterns.meta}
\usetikzlibrary{calc}
\tikzset{dotmark/.style={circle,fill,inner sep=1.5pt}}
\tikzset{emptymark/.style={circle,draw,fill=white,inner sep=1.5pt}}
\tikzset{crossmark/.style={thick,inner sep=1.5pt}}

\SetKwComment{Comment}{$\triangleright$\ \rm}{}
\SetKwInput{KwInput}{Input}
\SetKwInput{KwOutput}{Output}
\SetKwProg{Fn}{}{}{}

\newcommand{\A}{\mathcal{A}}

\def\ShowAuthNotes{1}
\ifnum\ShowAuthNotes=1
\newcommand{\authnote}[3]{\textcolor{#3}{[{\bf #1:} { {#2}}]}}
\else
\newcommand{\authnote}[3]{}
\fi

\newcommand{\Oh}{\mathcal{O}}
\newcommand{\Ohtilde}{\tilde{\Oh}}

\def\fragmentco#1#2{\bm{[}\,#1\,\bm{.\,.}\,#2\,\bm{)}}
\def\fragmentoc#1#2{\bm{(}\,#1\,\bm{.\,.}\,#2\,\bm{]}}
\def\fragmentoo#1#2{\bm{(}\,#1\,\bm{.\,.}\,#2\,\bm{)}}
\def\fragment#1#2{\bm{[}\,#1\,\bm{.\,.}\,#2\,\bm{]}}
\def\position#1{\bm{[}\,#1\,\bm{]}}

\newcommand{\ceil}[1]{\lceil #1 \rceil}

\newcommand{\yes}{{\sc yes}\xspace}
\newcommand{\no}{{\sc no}\xspace}

\renewcommand{\NG}{\operatorname{NG}}
\newcommand{\LG}{\operatorname{LG}}
\newcommand{\CNG}{\operatorname{CNG}}
\newcommand{\CLG}{\operatorname{CLG}}
\newcommand{\CG}{\operatorname{CG}}

\newcommand{\minscore}{\operatorname{min\ score}}

\newcommand{\one}{\mathtt{1}}
\newcommand{\zero}{\mathtt{0}}
\newcommand{\two}{\mathtt{2}}
\newcommand{\three}{\mathtt{3}}

\newcommand{\cost}{\mathsf{cost}}

\newcommand{\score}{\mathsf{score}}
\newcommand{\dyck}{\mathsf{dyck}}
\newcommand{\ed}{\mathsf{ed}}
\newcommand{\ted}{\mathsf{ted}}

\newcommand{\sub}{\mathsf{sub}}

\newcommand{\TED}{\textsf{TED}\xspace}
\newcommand{\DTED}{dynamic \textsf{TED}\xspace}

\newcommand{\ED}{\textsf{ED}\xspace}
\newcommand{\APSP}{\textsf{APSP}\xspace}
\newcommand{\OMv}{\textsf{OMv}\xspace}

\newcommand{\OPT}{\textsf{OPT}\xspace}
\newcommand{\Tminplus}{\textsf{T}_{\textsf{Min-Plus}}}
\newcommand{\TminplusOMv}{\textsf{T}_{\textsf{Min-Plus-OMv}}}

\newcommand{\ttA}{\mathtt{A}}
\newcommand{\ttB}{\mathtt{B}}
\newcommand{\ttR}{\mathtt{R}}
\newcommand{\ttL}{\mathtt{L}}
\newcommand{\ttC}{\mathtt{C}}

\newcommand{\mA}{\mathcal{A}}

\newcommand{\mX}{\mathcal{X}}
\newcommand{\mY}{\mathcal{Y}}

\newcommand{\mZ}{\mathcal{Z}}
\newcommand{\mI}{\mathcal{I}}

\newcommand{\mL}{\mathcal{L}}
\newcommand{\mS}{\mathcal{S}}

\newcommand{\bS}{\mathbf{S}}
\newcommand{\bG}{\mathbf{G}}

\newcommand{\bL}{\mathbf{L}}
\newcommand{\bR}{\mathbf{R}}

\newcommand{\bP}{\mathbf{P}}
\newcommand{\bT}{\mathbf{T}}

\newcommand{\w}{\mathrm{w}}

\renewcommand{\epsilon}{\varepsilon}

\DeclarePairedDelimiter\abs{\lvert}{\rvert}
\DeclarePairedDelimiter\norm{\lVert}{\rVert}

\makeatletter
\let\oldabs\abs
\def\abs{\@ifstar{\oldabs}{\oldabs*}}
\let\oldnorm\norm
\def\norm{\@ifstar{\oldnorm}{\oldnorm*}}
\makeatother

\def\problembox#1{%
    \vspace{2mm}%
    \noindent\fbox{%
    \begin{minipage}{.985\linewidth}%
        #1
    \end{minipage}%
    }%
    \vspace{2mm}%
}
\newcommand{\defproblem}[3]{%
    \problembox{%
        \textbf{#1}\\
        {\bf{Input:}} #2 ~\\
        {\bf{Output:}} #3
    }%
}

\makeatletter
\renewenvironment{cases}{%
  \matrix@check\cases\env@cases
}{%
  \endarray\right.%
}
\def\env@cases{%
  \let\@ifnextchar\new@ifnextchar
  \left\lbrace
  \def\arraystretch{1.1}%
  \array{@{\;}c@{\quad}l@{}}%
}
\makeatother

\def\mid{\ensuremath :}

\newcommand\thefont{\expandafter\string\the\font}

%% file: abstract.tex
String Edit Distance is a more-than-classical problem whose behavior in the dynamic setting, where the strings are updated over time, is well studied. 
A single-character substitution, insertion, or deletion can be processed in time $\tilde{\mathcal{O}}(n w)$\footnote{$\tilde{\mathcal{O}}(\cdot)$ hides polylogarithmic factors.} when operation costs are positive integers bounded by $w$ [Charalampopoulos, Kociumaka, Mozes, CPM 2020][Gorbachev, Kociumaka, STOC 2025]. If the weights are further uniform (insertions and deletions have equal cost), also an $\tilde{\mathcal{O}}(n \sqrt{n})$-update time algorithm exists [Charalampopoulos, Kociumaka, Mozes, CPM 2020]. This is a substantial improvement over the static $\mathcal{O}(n^2)$ algorithm when $w \ll n$ or when we are dealing with uniform weights. % In the unweighted setting, this update time is conditionally optimal under SETH [Backurs, Indyk, STOC 2015].

In contrast, for inherently related problems such as Tree Edit Distance, Dyck Edit Distance, and RNA Folding, it has remained unknown whether it is possible to devise dynamic algorithms with an advantage over the static algorithm. In this paper, we resolve this question by showing that (weighted) Tree Edit Distance, Dyck Edit Distance, and RNA Folding admit no dynamic speedup: under well-known fine-grained assumptions we show that the best possible algorithm recomputes the solution from scratch after each update. Furthermore, we prove a quadratic per-update lower bound for unweighted Tree Edit Distance under the $k$-Clique Conjecture. This provides the first separation between dynamic unweighted String Edit Distance and unweighted Tree Edit Distance, problems whose relative difficulty in the static setting is still open.

%% file: s1_intro.tex
\section{Introduction}

The Tree Edit Distance (\TED) measures the minimum number of insertions, deletions, and substitutions required to transform one labeled, ordered tree into another\footnote{More specifically, upon deleting a node, the children of the deleted node become the children of the original parent of the deleted node. Insertion is the reverse process of deletion.}, and has found numerous applications in settings where tree structures arise: computational biology
\cite{gusfield_1997,10.1093/bioinformatics/6.4.309,HochsmannTGK03,waterman1995introduction},
data analysis \cite{KochBG03,Chawathe99,FerraginaLMM09},
image processing \cite{BellandoK99,KleinTSK00,KleinSK01,SebastianKK04},
 compiler optimization \cite{DMRW10} and more.

 There is a long line of work spanning over thirty years \cite{Tai79, ShashaZhang89, Klein98, DMRW10} on improving the running time of \TED  to arrive at a cubic algorithm \cite{DMRW10}. It has only been shown recently that from a computational point of view,
weighted \TED (where there is a weight associated with every operation, and the goal is to minimize the total weight of the transformation) resembles the All-Pairs Shortest Paths (\APSP) Problem. Bringmann \emph{et al.} \cite{BGMW20} proved that weighted \TED cannot be solved in $\Oh(n^{3-\varepsilon})$ time under \APSP, and subsequently Nogler \emph{et al.} \cite{NPSVXY25} showed that it is, in fact, \emph{equivalent} to \APSP.
For unweighted \TED, however, the situation is drastically different and truly subcubic algorithms are known \cite{M22, Durr23, NPSVXY25}, with the best running time to date being ${\Ohtilde(n^{(3 + \omega)/2})} = {\Oh(n^{2.6857})}$ \cite{NPSVXY25}\footnote{$\omega < 2.372$ is the fast matrix multiplication exponent \cite{ADVXXZ24}.}. On the other hand, from a lower bound perspective all we know is that unweighted \TED is at least as hard as the String Edit Distance. It remains a major open problem whether unweighted \TED is strictly harder than String Edit Distance, conditioned under some fine-grained hardness hypothesis.

In this paper, we study the complexity of \TED (both weighted and unweighted) in the dynamic setting. While this question is well understood for the String Edit Distance problem, its complexity remains wide open for \TED. Note that the String Edit Distance, also known as the Levenshtein distance, is defined as the minimum number of single-character insertions, deletions, and substitutions required to transform one string into the other. When each operation is assigned a weight and the goal is to minimize the total cost, we obtain the \emph{weighted} edit distance. Computing the (weighted or unweighted) edit distance between two strings of length $n$ is a standard introductory example of dynamic programming, solvable in $\Oh(n^2)$ time \cite{Vintsyuk1968, NW70, WF74, Sel74}. Surprisingly, despite its simplicity, this algorithm is in fact optimal under the Strong Exponential Time Hypothesis (SETH) \cite{LBStringED15, LB2StringED15, LB3StringED15, LB4StringED15}.

The complexity of the Edit Distance Problem is well understood in the dynamic setting, where the strings undergo updates (insertions, deletions and substitutions) and we need to maintain the edit distance throughout the updates. After initial works that considered updates occurring only at the endpoints of the two strings \cite{LMS98, KP04, IIST05, Tis08}, Charalampopoulos, Kociumaka, and Mozes gave the first algorithm for the most general scenario \cite{CKM20}, where updates can happen anywhere. In the unweighted setting, their algorithm supports updates in $\tilde{\mathcal{O}}(n)$ time, and this is conditionally optimal under SETH. In the uniform-cost setting (where insertions and deletions have equal positive integer cost) bounded by $w = n^{\mathcal{O}(1)}$, their algorithm supports updates in $\tilde{\mathcal{O}}(n \cdot \min(\sqrt{n},w))$ time. Very recently, it has been shown that $\tilde{\mathcal{O}}(n \sqrt{n})$ is (conditionally) optimal when $w$ is a large polynomial in $n$ \cite{CKW23}. Moreover, an $\tilde{\mathcal{O}}(nw)$ update-time algorithm for general positive weights has been given \cite{GK24}. Beyond this, dynamic string Edit Distance has been also studied in more specific settings: Gorbachev and Kociumaka showed in \cite{GK24} that if the edit distance is \emph{bounded} by $k$, then the update time improves to $\Ohtilde(k)$. Moreover, for \emph{approximate} (unweighted) Edit Distance, there is an algorithm achieving subpolynomial amortized update time with a corresponding approximation guarantee \cite{KMS23}.

As stated earlier, it is a major open question whether unweighted \TED is strictly harder than unweighted string edit distance, and this same question remains in the dynamic setting.

\begin{center}
\textit{Question 1: Is unweighted \TED stricty harder than String Edit Distance?}
\end{center}

Very recently, \TED has been extensively studied in settings such as the approximate case \cite{BGHS19, Seddighin22} and the bounded case \cite{DGHKSS22, DGHKS23, KS25}, where several breakthrough results have been achieved.
However, almost nothing is known about its behavior in the dynamic setting. The only (very recent) work about \TED that supports update is a  dynamic algorithm with approximation factor $n^{1/2+o(1)}$ and subpolynomial update time \cite{DGKHS25}. However, this result is obtained through an embedding into String Edit Distance and employing known algorithms for strings, rather than an algorithm specific to trees.

The lack of knowledge for dynamic (exact) \TED stands out even more jarringly when contrasted with our good understanding of dynamic \APSP \cite{Kin99, KT01, DI02, DI04, Tho04, Tho05, ACK17, GWN20b, CZ23, Mao24} (which, in the static setting, is equivalent to weighted \TED) and dynamic String Edit Distance (which may still be no easier than unweighted \TED).
In fact, it is reasonable to assume from the recent work of \cite{NPSVXY25} that dynamic improvements for \TED might be possible: The authors introduce the notion of border-to-border distance on tree-based graphs. It is a generalization of the \emph{string} alignment graph, which is used by \cite{CKM20} in their dynamic algorithm for String Edit Distance. When we represent the  strings as two sequences of singleton trees, we recover exactly the string alignment graph. The seeming connection between the alignment graphs raises the question:

\begin{center}
\textit{Question 2: Is there a dynamic advantage for dynamic Tree Edit Distance?}
\end{center}

\subparagraph{Dyck Edit Distance and RNA Folding.}
Another problem closely related to String Edit Distance is the Dyck Edit Distance Problem, which asks for the minimum number of edits needed to transform a given (not necessarily balanced) parentheses string into a well-balanced one.
A further related problem is the RNA Folding Problem: given a string over the nucleotides \texttt{A}, \texttt{U}, \texttt{C}, and \texttt{G}, the goal is to find a maximum set of non-crossing base pairs (\texttt{A},\texttt{U} and \texttt{C},\texttt{G}) so that the string “folds” into a valid secondary structure.
(See \cref{sec:preliminaries} for formal definitions and a discussion of why these problems are so closely connected.)

In the static setting, both problems allow for a simple $\Oh(n^3)$ dynamic programming solution, assuming constant grammar size for Dyck Edit Distance \cite{aho1972minimum, NJ80}. A series of (inexhaustive) works \cite{VGF14, BGSV19, DBLP:conf/soda/WilliamsX20, CDX22} improves this running time to the current best of $\Ohtilde(n^{(3+\omega)/2})$ via bounded-monotone min-plus products \cite{CDXZ22}. A lot of attention has also been given to the approximation setting \cite{s14, s15, DKS22, KS23, BGHS19,Seddighin22}.
In terms of hardness, Abboud \emph{et al.} \cite{ABVW15} show a \emph{combinatorial} $\Omega(n^3)$ lower bound for both problems. However, this does not exclude the possibility of a faster dynamic (combinatorial) algorithm, as such an algorithm would be highly desirable in practice.
However, in the dynamic setting, nothing is known of these two problems except for an approximation algorithm \cite{DGKHS25} for dynamic Dyck Edit Distance, which relies on an embedding into strings. This leads to our third question: 

\begin{center}
\textit{Question 3: Is there a dynamic advantage for dynamic Dyck Edit Distance and RNA Folding?}
\end{center}

\paragraph*{Our Results.}

In our paper, we make progress in answering the above questions.
Technique-wise, we use static hardness conjectures to prove lower bounds for dynamic problems, a technique started by \cite{popular14}.
We base our hardness on the Min-Weight $k$-Clique Conjecture and the $k$-Clique Detection Conjecture in fine-grained complexity (see \cite{finegrainedsurvey} for background). The two problems are defined as follows:

\defproblem
{Min-Weight $k$-Clique}
{A weighted graph $\bG = (V, E, w)$, where $V = \fragment{1}{n}$.}
{$\min_{v_1, \ldots, v_k \in V} \sum_{i < j } w(v_i, v_j)$ such that $v_1, \ldots, v_k$ is a $k$-clique.} 
\defproblem
{$k$-Clique Detection}
{A unweighted graph $\bG = (V, E)$ on $n$ nodes, where $V = \fragment{1}{n}$.}
{\yes if there are $v_1, \ldots, v_k \in V$ such that $v_1, \ldots, v_k$ is a $k$-clique, and \no otherwise.}

\begin{conjecture} [Min-Weight $k$-Clique Conjecture] \label{conj:weightedclique}
    For any $\varepsilon > 0$, there exists a constant $c > 0$ such that for any $k \geqslant 3$, the Min-Weight $k$-Clique Problem with edge weights in $\{1, \dots, n^{ck}\}$ cannot be solved in $\Oh(n^{k(1-\varepsilon)})$ time.
\end{conjecture}

\begin{conjecture} [Combinatorial $k$-Clique Detection Conjecture] \label{conj:cliquedetection}
    For any $\varepsilon > 0$, the $k$-Clique Detection Problem cannot be solved in $\Oh(n^{k - \varepsilon})$ time by any combinatorial algorithm.
\end{conjecture}

The term ``combinatorial algorithms'' has been used very frequently by the research community (including but not limited
to \cite{HKN13, RT13, popular14, HenzingerKNS15, VW18, AFKLM24}), although it is not a well-defined term. In our paper, as well as many other papers, it is used to refer to algorithms that do not use ``algebraic'' techniques such as in Strassen's algorithm \cite{Strassen1969}, which underlie all known algorithms for fast matrix multiplication. (This definition has been first introduced by Ballard \emph{et al.}  \cite{BDHS13}.) Compared to algebraic algorithms, combinatorial algorithms are generally more feasible to implement in practice. Thus, a lower bound for combinatorial algorithms can be somehow also considered as a ``practical'' lower bound.

For the first question, we show that unweighted \TED is strictly harder than String Edit Distance in the dynamic setting when restricted to combinatorial algorithms. 
\begin{restatable*}{mtheorem}{unweightedted}\label{thm:unweightedted}
Unless \cref{conj:cliquedetection} fails, for any $\varepsilon > 0$, there is no combinatorial algorithm for (unweighted) \DTED of size $\Oh(N)$ that satisfies $p(N) + N (u(N) + q(N)) = \Oh(N^{3-\varepsilon})$,
where $p(N), u(N)$, and $q(N)$ denote the preprocessing, update, and query times of the algorithm, respectively.
This lower bound holds even when the alphabet size is constant.
\end{restatable*}

Thus, even though it still remains open whether unweighted \TED can be solved in the same time complexity as String Edit Distance,
we show that at least in the dynamic setting, these two problems are separate.

For the second question, we show that there is no dynamic advantage to weighted Tree Edit Distance. 
\begin{restatable*}{mtheorem}{weightedted}\label{thm:weightedted}
Unless \cref{conj:weightedclique} fails, for any $\varepsilon > 0$, there is no algorithm for \DTED of size $\Oh(N)$ and alphabet size $\Oh(N)$ that satisfies $p(N) + N (u(N) + q(N)) = \Oh(N^{4-\varepsilon})$,
where $p(N), u(N)$, and $q(N)$ denote the preprocessing, update, and query times of the algorithm, respectively.
\end{restatable*}

Thus, while both weighted String Edit Distance and \APSP admit faster dynamic advantage, weighted \TED does not: under the Min-Weight $k$-Clique Conjecture, it cannot be improved beyond recomputing from scratch after each update. This highlights a major difference between weighted \TED and the two problems which it is most closely related to, both computationally and definition-wise. In fact, our lower bound holds even for incremental and decremental \TED.

To answer the third question, we show that there is no dynamic advantage to RNA Folding and Dyck Edit Distance among combinatorial algorithms.

\begin{restatable*}{mtheorem}{rnafolding}\label{thm:rnafolding}
Unless \cref{conj:cliquedetection} fails, for any $\varepsilon > 0$, there is are no combinatorial algorithms for dynamic RNA Folding or dynamic Dyck Edit Distance of size $\Oh(N)$ that satisfy $p(N) + N (u(N) + q(N)) = \Oh(N^{4-\varepsilon})$,
where $p(N), u(N)$, and $q(N)$ are the preprocessing, update, and query times of the algorithm, respectively. This lower bound holds even when the alphabet size is constant.
\end{restatable*}

Finally, we remark that by combining existing results, one can obtain a dynamic advantage for RNA Folding and Dyck Edit Distance when the updates are only insertions at the end of the string (essentially placing us in an online setting). By combining \cite{DG24} and \cite{HP25}, one can show an even stronger result: in this online setting, these problems achieve the same time complexity as Online Matrix-Vector Multiplication (\OMv) Problem \cite{HenzingerKNS15}. The latter is one of the most recent studied problems in fine-grained complexity for dynamic lower bounds: given an $n \times n$ matrix $M$ and a sequence of vectors $v_1, \ldots, v_n$ arriving one by one, the task is to output $Mv_i$ before the arrival of $v_{i+1}$.

\begin{restatable*}{lemma}{onlineRNA}\label{thm:online}
    There is a randomized algorithm that solves online RNA Folding and online Dyck Edit Distance of size $\Oh(N)$ in the same total time complexity as \OMv (currently $N^3/2^{\Omega(\sqrt{\log(N)})}$ by \cite{OMVwilliams}),  and succeeds with high probability.
\end{restatable*}

%% file: s2_overview.tex
\section{Overview}

\subparagraph*{Lower Bound for Dynamic Unweighted \TED.}

We show how to reduce an instance of $3k$-Clique Detection, $G$, to an instance of dynamic unweighted \TED.
The key ingredient is the construction of embeddings $\CLG$ and $\CNG$, referred to as \emph{clique gadgets}, which map two $k$-cliques $X, Y$ in $\bG$ into strings $\CLG(X)$ and $\CNG(Y)$ of length $\Oh(n \log n)$. These gadgets are designed so that $\ed(\CLG(X), \CNG(Y)) = C$ for some fixed constant $C$ if and only if $X$ and $Y$ are fully connected; otherwise, the edit distance is strictly larger than $C$.

First, we enumerate the set $\mS$ of all $k$-cliques in $\bG$.
Then, for each fixed $Z \in \mS$, we construct two trees $\bT(Z)$ and $\bT'(Z)$ that satisfy:
\begin{enumerate}
\item The tree edit distance $\ted(\bT(Z), \bT'(Z))$ minimizes, over all $X, Y \in \mS$, the sum
$\ed(\CLG(X), \CNG(Y)) + \ed(\CNG(X), \CLG(Z)) + \ed(\CLG(Y), \CNG(Z))$.
\item The sizes of $\bT(Z)$ and $\bT'(Z)$ are $\Oh(n^{k+1} \log n)$, while the part of the construction that depends on $Z$ involves only $\Oh(n \log n)$ nodes in $\bT(Z), \bT'(Z)$.
\end{enumerate}

Next, we run $\abs{\mS} = \Oh(n^k)$ rounds, one for each $Z \in \mS$. In each round, we maintain $\bT(Z),\bT'(Z)$, update them to the next instance using only $\Oh(n \log n)$ updates, and then query for the tree edit distance. This allows us to solve the $3k$-Clique Detection instance in $\Oh(n^{k+1} \log n)$ updates and $\Oh(n^{k} \log n)$ queries. A careful analysis of the factors, together with choosing $k$ sufficiently large, yields \cref{thm:unweightedted}.

\subparagraph*{Lower Bound for Dynamic Weighted \TED.}
Our starting point is the \TED lower bound construction of Bringmann \emph{et al.} \cite{BGMW20}, which encodes a Minimum Weight Triangle instance (known to be equivalent to \APSP) into a \TED instance. We extend such a reduction into another one that simultaneously ensures the two following:
\begin{enumerate}
\item instead of yielding the minimum weight triangle, the instance should yield the minimum weight 4-clique with one of the four vertices fixed; and
\item the portion of the instance that depends on the chosen fixed vertex should involve only $\Oh(1)$ nodes.
\end{enumerate}
This allows us to proceed similarly to the unweighted case. 
Given a 4-clique instance $\bG$, the algorithm runs for $n$ rounds, one for each vertex in $\bG$. In each round, we maintain the instance corresponding to the fixed vertex, update it to the next instance using only $\Oh(1)$ updates, and then query for the minimum-weight 4-clique containing the current vertex. Taking the minimum over all queried values yields the minimum weight 4-clique.

We remark that although, at a high level, this reduction may appear similar to the one for the unweighted case, the unweighted instances $\bT(Z)$ and $\bT(Z')$ differ significantly from the weighted \TED instances (compare \cref{fig:unweighted_instance} with \cref{fig:weighted:b}). The latter rely heavily on extremely large weights, whereas the construction of $\bT(Z)$ and $\bT(Z')$ is closer in spirit to gadgets used for unweighted string edit distance such as \cite{LBStringED15} than to those for weighted \TED.

\subparagraph*{Lower Bound for Dynamic RNA folding and Dyck Edit Distance.} Our starting point is once again the static lower bound \cite{ABVW15}. Conceptually, the reduction is similar to the previous ones: the static lower bound encodes a $3k$-Clique Detection instance into a single RNA Folding instance. For the dynamic lower bound, however, we instead encode a $4k$-Clique Detection instance into a single RNA Folding instance, where $k$ of the vertices are fixed for each round. Similar to the previous reductions, 
we must ensure that the part of our dynamic RNA Folding instance depending on the currently fixed vertices is small in size, so the number of updates stays small.
Regarding Dyck Edit Distance, we do not need to build any new instance. We use the embedding of RNA Folding into Dyck Edit Distance from \cite{C15}.

\subparagraph*{Organization of the Paper.}
First, we set up notation in \cref{sec:preliminaries}.
\cref{sec:unweighted}, \cref{sec:weighted} and \cref{sec:rnafolding} are dedicated to the proof of \cref{thm:unweightedted}, \cref{thm:weightedted}, and \cref{thm:rnafolding}, respectively. For sake of completeness, we also present the proof of \cref{thm:online} in \cref{sec:appendix}.

%% file: s3_prelims.tex
\section{Preliminaries}
\label{sec:preliminaries}
\subparagraph*{Sets.}
For integers \(i, j \in \mathbb{Z}\), we write \(\fragment{i}{j}\) to represent the set \(\{i, \dots, j\}\), and \(\fragmentco{i}{j}\) to denote the set \(\{i, \dots, j - 1\}\).
We define \(\fragmentoc{i}{j}\) and \(\fragmentoo{i}{j}\) similarly.

\subparagraph*{Strings.} 
An \emph{alphabet} $\Sigma$ is a finite set of symbols. A \emph{string} $X \in \Sigma^n$ of length $n$ is written as $X = X\position{1} X\position{2} \cdots X\position{n}$, where each $X\position{i}$ denotes the $i$-th character of $X$, for $i \in \{1, \dots, n\}$. The length of string $X$ is denoted as $|X|$.
Given indices $1 \le i \le j \le |X| + 1$, we define the \emph{fragment} $X\fragmentco{i}{j} := X\position{i} \cdots X\position{j - 1}$. A \emph{prefix} of $X$ is any fragment starting at position $1$, and a \emph{suffix} is any fragment ending at position $|X|$.
A string $Y$ of length $m \in \fragment{0}{n}$ is a \emph{substring} of $X$ if there exists an index $i$ such that $Y = X\fragmentco{i}{i + m}$.
For two strings $A$ and $B$, their \emph{concatenation} is denoted by $A \circ B$ or simply $AB$. The notation $A^k$ represents the string formed by concatenating $k$ copies of $A$. The \emph{reverse} of a string $X$ is written as $X^R := X\position{n} \cdots X\position{1}$.

Finally, for strings $A$ and $B$, we denote by $\ed(A, B)$ the \emph{(string) edit distance} (in short, \ED) between $A$ and $B$, defined as the minimum number of character insertions, deletions and substitutions to transform $A$ into $B$.

\subparagraph*{Tree Edit Distance (\TED).} We consider ordered, node-labeled trees from an alphabet $\Sigma$. Given a node $v$ in a tree, we use $\ell(v)$ to denote its label, and we use $\sub(v)$ to denote the subtree rooted at $v$.

\defproblem
{Weighted Tree Edit Distance}
{Two trees $\bT, \bT'$ labeled from an alphabet $\Sigma$ and a cost function $\delta:\Sigma \cup \{\varepsilon\} \times \Sigma \cup \{\varepsilon\} \mapsto \mathbb{R}$.}
{$\ted(\bT, \bT')$ defined as the minimum cost of transforming $\bT$ into $\bT'$ by performing a sequence of edit operations, which consist of the following three types: 
\begin{itemize}
        \item Changing the label of a node from $\ell$ to $\ell'$, $\ell \neq \ell'$ at cost $\delta(\ell, \ell')$.
        \item Deleting a node $v$ and attaching its children (if there are any) to the parent of $v$ in their original order (if $v$ is the root, then we obtain an ordered forest) at cost $\delta(\ell(v),\varepsilon)$.
        \item For an existing node $v'$ inserting a new node $v$ as a new leaf child in some position, or selecting consecutive children $w_1, \dots, w_k$ of $v'$, inserting a new node $v$ as the new parent of $w_1, \dots, w_k$, and placing $v$ as a child of $v'$ at cost $\delta(\epsilon, \ell(v))$.
\end{itemize}}

When all edit operations cost 1, we say the tree edit distance is \emph{unweighted}.
For our paper, it is more convenient to work with a different definition of \TED, which is more similar to the definition used in \cite{M22} and \cite{NPSVXY25}, and is defined through the notion of \emph{tree-alignment}.

\begin{definition} \label{def:tree-alignment}
   For two trees $\bT,\bT'$,
   we say the sequence $\mA = \{(v_i, v_i')\}_{i=1}^k \in \bT \times \bT'$ 
   is a \emph{tree-alignment of $\bT$ onto $\bT'$} if for all $1 \leq i < j \leq k$:
   \begin{itemize}
       \item $v_i$ is an ancestor of $v_j$ in $\bT$ iff $v_i'$ is an ancestor of $v_j'$ in $\bT'$,
       
       \item if neither $v_i$ nor $v_j$ is the ancestor of the other, then $v_i$ comes before $v_j$ in the pre-order traversal of $\bT$
       iff $v_i'$ comes before $v_j'$ in the pre-order traversal of $\bT$. \qedhere
    \end{itemize}
\end{definition}

For a tree-alignment $\mA = \{(v_i, v_i')\}_{i=1}^k \in \bT \times \bT'$ and $i \in \fragment{1}{k}$, we say $\mA$ \emph{aligns $v_i$ with $v_i'$}.
If further $\ell(v_i) = \ell(v_i')$, then we say $\mA$ \emph{matches} $v_i$ with $v_i'$,
otherwise we say $\mA$ \emph{substitutes} $v_i$ with $v_i'$.
The \emph{cost} of a tree-alignment $\mA$, denoted by $\cost(\mA)$,
is defined as 
\[
    \cost(\mA) \coloneqq \sum_{\substack{(v, v') \in \mA \\ \ell(v) \neq \ell(v')}} \delta(\ell(v),\ell(v')) + \sum_{v \in \bT \mid \nexists_{v'} (v,v') \in \mA} \delta(\ell(v), \varepsilon) + \sum_{v' \in \bT' \mid \nexists_{v} (v,v') \in \mA} \delta(\varepsilon, \ell(v')).
\]
Finding the minimum cost of a tree-alignment is an equivalent definition of \TED \cite{M22}\footnote{To be precise, in \cite{M22} and \cite{NPSVXY25}, the complementary problem is considered, i.e., finding the maximum value of $\sum_{v \in \bT} \delta(\ell(v), \varepsilon) + \sum_{v' \in \bT'} \delta(\varepsilon, \ell(v')) - \cost(\mA)$, taken over all tree-alignments $\mA$. This is evidently equivalent to minimizing $\cost(\mA)$, in which case we recover the tree edit distance.}.

In our paper, we also view String Edit Distance through the lens of \emph{string-alignment}. In this case, we can think of each string as a forest consisting of $n$ single-node trees. Then \cref{def:tree-alignment} recovers exactly their string-alignment.

\subparagraph{Path Gadgets.}

To describe more concisely our lower bound instances for \TED, we define a basic gadget \( \bP(S) \), which, given a string \( S = s_1 \ldots s_d \), constructs a path of length \( d \) with nodes labeled \( s_1, \ldots, s_d \) from top to bottom.
Moreover, given a path \( \bP \) and a node \( v \in \bP \), we say that we \emph{right-attach} a string \( S = s_1, \ldots, s_d \) to \( v \) when we attach \( d \) nodes as right children of \( v \), labeled \( s_1, \ldots, s_d \) from left to right.  
We define \emph{left-attach}ing analogously, for attaching nodes to the left of \( v \).

\subparagraph{RNA Folding.}

The RNA Folding Problem is defined over strings on the alphabet \( \Sigma \cup \Sigma' \),  where each character in \( \Sigma \) has a correspondingly paired character in \( \Sigma' \).  
For \( \sigma \in \Sigma \), we denote its paired character in \( \Sigma' \) by \( \sigma' \). Moreover, for a string \( S \) over \( \Sigma \cup \Sigma' \),  
we define \( p(S) \) to be the string obtained by replacing each character with its paired character.

Two index pairs \( (i, j) \) and \( (i', j') \), where \( i < j \) and \( i' < j' \), are said to \emph{cross} if either $i < i' < j < j'$ or $i' < i < j' < j$ holds. Now, given a string \( S \in (\Sigma \cup \Sigma')^n \), a set of index pairs  
\( F = \{(i, j) \mid 1 \le i < j \le n\} \) is a \emph{folding of \( S \)} if:
\begin{enumerate*}[(i)]
    \item for all distinct \( (i, j), (i', j') \in F \), the pairs do not cross; and
    \item for all \( (i, j) \in F \), we have \( S[i] = p(S[j]) \).
\end{enumerate*}
This allows us to define the RNA Folding Problem, along with a weighted variant.

\defproblem
{RNA Folding}
{A string $S \in \{\Sigma \cup \Sigma'\}^n$.}
{$\score(S) \coloneqq \max_{F \text{ folding of \( S \)}} |F|$.}

\defproblem
{Weighted RNA Folding}
{A string $S \in \{\Sigma \cup \Sigma'\}^n$ and a weight function $\w : \Sigma \cup \Sigma' \rightarrow \fragment{1}{M}$ where $w(\sigma) = w(p(\sigma))$ for all $\sigma \in \Sigma \cup \Sigma'$.}
{$\score_\w(S) \coloneqq \max_{F \text{ folding of \( S \)}} \sum_{(i,j) \in F}\w(S\position{i})$.}
The weighted version is convenient because it can be reduced to the unweighted problem with only an overhead of \( M \) on the length of the new string, where \( M \) is the largest weight. This allows us to work with the weighted formulation, as long as we keep $M$ reasonably small. 

\begin{lemma}[\cite{ABVW15}] \label{lem:weighted_rna_folding}
    Let $S \in \{\Sigma \cup \Sigma'\}^n$ be a string and $\w : \Sigma \rightarrow \fragment{1}{M}$ a weight function.
    
    Then, $\score_{\w}(S) = \score(S')$ for the string $S' \coloneqq s_1^{\w(s_1)} \cdots s_n^{\w(s_n)}$.
    \lipicsEnd
\end{lemma}

\subparagraph{Dyck Edit Distance.}

The Dyck Edit Distance Problem asks for the minimum number of edits required to transform a given string into a well-balanced string of parentheses. It can be thought of as a variant of RNA Folding, with three key differences: 
\begin{enumerate*}[(1)]
    \item it is formulated as a minimization problem rather than a maximization problem;
    \item it allows also for substitution; and
    \item a symbol $\sigma$ may be matched with its corresponding closing symbol $\sigma'$ only if $\sigma$ appears before $\sigma'$ in the string, but not the other way around.
\end{enumerate*}

\defproblem
{Dyck Edit Distance.}
{A string $S \in \{\Sigma \cup \Sigma'\}^n$.}
{$\dyck(S) \coloneqq \min_{S' \in \mL_{\dyck}(\Sigma)} \ed(S, S')$, where $\mL_{\dyck}(\Sigma)$ is the language defined by the grammar with the rules $\bS \rightarrow \bS\bS$, $\bS \rightarrow \varepsilon$, and $\bS \rightarrow \sigma \bS \sigma'$ for all $\sigma \in \Sigma$.}

Regarding Dyck Edit Distance, no additional setup is required: for the purpose of proving lower bounds, one can directly reduce RNA Folding into Dyck Edit Distance.

\begin{lemma}[\cite{C15}] \label{lem:dyck_to_rnafolding}
    Let $S = s_1 \cdots s_n$ be a string over $\Sigma \cup \Sigma'$.
    Let $\bar{\Sigma} \coloneqq \{\#\} \cup \{\zero_{\sigma},  \zero_{\sigma}', \one_{\sigma},  \one_{\sigma}' \}_{\sigma \in \Sigma}$,
    and define the function $\phi : \Sigma \rightarrow {\bar{\Sigma}^{6}}$ such that 
    $\phi(\sigma) = \zero_{\sigma} \ \# \ \one_{\sigma}' \ \zero_{\sigma} \ \# \ \one_{\sigma}'$ and $\phi(\sigma') = \one_{\sigma} \ \one_{\sigma}  \ \zero_{\sigma}' \  \zero_{\sigma}'$, 
    for $\sigma \in \Sigma$.
    
    Then, for $\Phi(S) \coloneqq \phi(s_1) \cdots \phi(s_n)$,
    we have $\dyck(\Phi(S)) = 3|S| - 2\score(S)$\footnote{In \cite{C15}, the author specifically considers the case $|\Sigma| = 2$, but its proof extends straightforwardly to alphabets of arbitrary size $|\Sigma| > 2$ by symmetry, as long as distinct letters are used for each symbol in the embedding.}.
    \lipicsEnd
\end{lemma}

\subparagraph{Clique Gadgets.}
Central to the lower bound construction of \cite{ABVW15}, and also to the dynamic lower bounds presented here for Dyck Edit Distance, RNA Folding, and unweighted \TED, are the gadget constructions from \cite{ABVW15} that encode the neighborhoods of nodes and set of nodes. 
To describe such gadgets formally, given a graph $\bG = (V,E)$ and a node $v \in V$, we use $N(v)$ to denote the set of neighbors of $v$. For a subset $V' \subseteq V$, we define the neighbor set of $V'$ to be $N(V') = \bigcap_{v' \in V'}N(v')$.

\begin{lemma}[Claim 4 in \cite{ABVW15}]\label{lem:clique_gadged}
Let $\bG = (V, E)$ be a graph with $\abs{V} = n$, and let $k \in \mathbb{Z}_{+}$.
Then, there exist two string embeddings $\CLG : V^k \rightarrow \Sigma^{\lambda}$ and $\CNG : V^k  \rightarrow \Sigma^{\lambda}$ of length\footnote{In \cite{ABVW15}, the embeddings have length \emph{at most} $\lambda$. We pad each one with two fresh symbols so that every embedding now has length \emph{exactly} $\lambda$, without affecting the score.} $\lambda = \Oh(nk \log n)$ over an alphabet of size $|\Sigma| = \Oh(1)$, and a constant $C = C(n, k)$ such that for any $X,Y \in V^k$:
\begin{align*}
    & \score(\CLG(X) \circ p(\CNG(Y))^R) = C && \text{if $X \subseteq N(Y)$,}\\
    & \score(\CLG(X) \circ p(\CNG(Y))^R) < C && \text{otherwise.}
\end{align*} 
\lipicsEnd
\end{lemma}

Thus, these gadgets allow us to determine whether two $k$-cliques are fully connected to each other.

%% file: s4_unweighted.tex
\section{Lower Bounds for Dynamic Unweighted \TED}
\label{sec:unweighted}

In this section, we prove \cref{thm:unweightedted}.
As a first step, we construct gadgets that serve the same purpose as those in \cite{ABVW15}, but are adapted to work under edit distance.

\begin{lemma}\label{cor:clique_gadged}
Let $\bG = (V, E)$ be a graph, and let $k \in \mathbb{Z}_{+}$.
Then, there exist two string embeddings $\CLG : V^k \rightarrow \Sigma^{\lambda_1}$ and $\CNG : V^k  \rightarrow \Sigma^{\lambda_2}$ of lengths $\lambda_1, \lambda_2 = \Oh(nk^3 \log n)$ over a alphabet of size $|\Sigma| = \Oh(1)$, and a constant $C = C(n, k)$ such that for any $X,Y \in V^k$:
\begin{align*}
    & \ed(\CLG(X), \CNG(Y)) = C && \text{if $X \subseteq N(Y)$,}\\
    & \ed(\CLG(X), \CNG(Y)) > C && \text{otherwise.}
\end{align*}
\end{lemma}

\begin{proof}
    In this proof, we use the alphabet $\Sigma \coloneqq \{\zero, \one, \two, \three, \$, \#\}$.
    We set $\alpha = \Oh(\log n)$ to be the number of bits needed to encode values of $V = \fragment{1}{n}$ in binary. Moreover, for a number $x \in V$, we denote with $\bar{x} \in \{\zero, \one\}^{\alpha}$ its binary encoding.

    For nodes $v,u \in V = \fragment{1}{n}$, 
    we define the \emph{node gadget} $\NG(v)$
    and \emph{list gadget} $\LG(v)$ as
    \begin{center}
        $\NG(v) \coloneqq \left( \bigcirc_{i=1}^{n-1} \two^{\alpha} \ \$ \right) \ \bar{v} \ \$ \ \left( \bigcirc_{i=1}^{n-1} \two^{\alpha} \ \$ \right)$
        \quad and \quad  
        $\LG(u) \coloneqq \left( \bigcirc_{w=1}^{n} \ g_u(w) \ \$ \right)$,
    \end{center}
    where $g_u(w) = \bar{w}$ if $w \in N(u)$ and $g_u(w) = \three^{\alpha}$, otherwise.

    \begin{claim}\label{claim:clique_gadget:1}
        Let $C' \coloneqq (n-1)(2\alpha+1)$. Then, for nodes $u,v \in V$, we have:
        \begin{align*}
            & \ed(\NG(v), \LG(u)) = C' && \text{if $v \in N(u)$,}\\
            & \ed(\NG(v), \LG(u)) \in \fragment{C'+1}{C'+\alpha} && \text{otherwise.} 
        \end{align*}
    \end{claim}
    \begin{claimproof}
        We first prove that $\ed(\NG(v), \LG(u)) \leq C' + \alpha$ (and if $v \in N(u)$, then $\ed(\NG(v), \LG(u)) \leq C'$).
        To this end, construct the following (string) alignment $\mA_u$:
        \begin{enumerate}[(a)]
        \item Substitute $g_u(v)$ with $\bar{v}$ (if $v \in N(u)$ then match exactly as $g_u(v) = \bar{v}$);
        \label{it:ed_gadget:a}
        \item Align the $(\alpha+1)(v-1)$ characters coming before $g_u(v)$ and $\bar{v}$ to each other;
        \label{it:ed_gadget:b}
        \item Align the $(\alpha+1)(n -v)$ characters coming after $g_u(v)$ and $\bar{v}$ to each other; and
        \label{it:ed_gadget:c}
        \item delete the remaining characters.
        \label{it:ed_gadget:d}
        \end{enumerate}

         We briefly argue that $\cost(\mA_u) \leq C' +\alpha = (n-1)(2\alpha+1) + \alpha$:
        in \ref{it:ed_gadget:a}, we incur cost at most $\alpha$ (this becomes $0$ if $v \in N(u)$), in \ref{it:ed_gadget:b} and \ref{it:ed_gadget:c}
        we incur cost $\alpha (v - 1)$ and $\alpha (n - v)$, respectively,
        as out of the $(\alpha+1)(v-1)$ and $(\alpha+1)(n -v)$ characters, there are $v-1$ and $n-v$
        perfectly matched $\$$. Finally, in \ref{it:ed_gadget:d} we incur cost $(\alpha+1)(n-1)$,
        as the remaining characters include a length-$(\alpha+1)(n-v)$ prefix and length-$(\alpha +1)(v-1)$ suffix
        of $\NG(u)$.
        
        To conclude the proof, we consider an optimal alignment $\mA$,
        and show that it has cost at least $C'$.
        Furthermore we show that whenever $\mA$ has cost exactly $C'$, then $v \in N(u)$.
        
        To this end, let $m$ be the number of positions $i$ in $\LG(u)$ such that $\LG(u)\position{i}$ is not matched perfectly by $\mA$. Note that $m \geq \alpha(n-1)$, as out of the $n\alpha$ many non-$\$$ characters in $\LG(u)$, no more than $\alpha$ many $\zero/\one$ characters can be matched exactly.
        Moreover, from $\abs{\NG(v)} \geq \abs{\LG(u)}$ follows that the cost of $\mA$ is at least
        \[
            (\abs{\NG(v)} - \abs{\LG(u)}) + m \geq (\alpha+1)(2n-1) - (\alpha+1)n + \alpha(n-1) = (n-1)(2\alpha+1) = C'.
        \]
        Now, if $\mA$ has cost exactly $C'$, this last inequality must be tight and $m = \alpha(n-1)$.
        In particular, this means that $\mA$ matches perfectly all $\$$ characters and $\alpha$ many $\zero/\one$
        characters of $\LG(u)$.
        Observe that these $\alpha$ many $\zero/\one$ characters must come from the same $g_u(w)$, for some $w \in V$.
        (If not, then at least one $\$$ would be deleted, contradicting the fact that all $\$$ of $\LG(v)$ are perfectly matched.)
        The only place where $g_u(w)$ can be matched perfectly to is $\bar{v}$, as it is the only part of $\NG(v)$ having $\zero/\one$ characters. Since $g_u(w)$ and $\bar{v}$ are matched perfectly, we have $g_u(w) = \bar{v}$, which can only happen when $w = v$ and $w \in N(u)$, i.e., $v \in N(u)$.
    \end{claimproof}

    Next, for subsets $X,Y \subseteq V^k$, we define the \emph{clique node gadget} $\CNG(Y)$
    and \emph{clique list gadget} $\CLG(X)$ as 
    \begin{center}
        $\CNG(Y) \coloneqq \bigcirc_{v \in Y} {\left( \ \NG(v) \ \#^{\ell} \ \right)}^k$
        \quad and \quad  
        $\CLG(X) \coloneqq {\left( \bigcirc_{v \in X} \ \LG(v) \ \#^{\ell} \right)}^k$,
    \end{center}
    where $\ell \coloneqq \ceil{k^2 \cdot (C'+\alpha)/2}$.
    Thus, $\lambda_1, \lambda_2$ satisfy $\lambda_1, \lambda_2 = \Oh(k \cdot \ell) = \Oh(nk^3 \log n)$.

    For $i \in \fragment{1}{k^2}$, let $B_{i}$ and $B_i'$ be the $i$-th substrings of the form $\NG(v)$
    and $\LG(v)$ appearing in $\CNG(Y)$ and $\CLG(X)$, respectively. Moreover, let $H_i$ and $H_i'$
    be the substrings $\#^{\ell}$ right after $B_i$ and $B_i'$.
    Lastly, let $\mA$ be an optimal (string) alignment that aligns $\CLG(X)$ onto $\CNG(Y)$
    and maximizes the number of $i,j$ such that $H_i'\position{j}$ is matched to $H_i\position{j}$.

    \begin{claim}\label{claim:clique_gadget:2}
        There are no $i,j \in \fragment{1}{k^2}$ such that $i \neq j$ and $\mA$ aligns any character of $B_i$ and $B'_j$ to each other.
    \end{claim}
    \begin{claimproof}
        For the sake of contradiction, assume there exist indices $i \neq j$ such that a character from $B_i$ matches a character from $B'_j$.
        Observe that $B_i$ is preceded by $(i-1)\ell$ and followed by $(k^2-i-1)\ell$ occurrences of $\#$, while $B_j'$ is preceded by $(j-1)\ell$ and followed by $(k^2-j-1)\ell$ occurrences of $\#$.
        Since $i \neq j$, alignment $\mA$ necessarily incurs a cost of at least $\ell$ both before and after the matched character between $B_i$ and $B'_j$. Hence, $\cost(\mA) \geq 2\ell$.

        Now consider an alternative alignment $\mA'$ which matches every $H_x$ perfectly with $H'_x$, and aligns each $B_x$ optimally with $B'_x$ for all $x \in \fragment{1}{k^2}$.
        By \cref{claim:clique_gadget:2}, the cost of such an alignment satisfies $\cost(\mA') \leq k^2 \cdot (C' + \alpha) \leq \cost(\mA)$. This contradicts the assumed optimality of $\mA$ (and if $\cost(\mA') = \cost(\mA)$, then $\mA'$ is strictly better in terms of the number of pairs $(x,y)$ where $H'_x\position{y}$ is matched to $H_x\position{y}$).
    \end{claimproof}

    \begin{claim}\label{claim:clique_gadget:3}
        For all $i \in \fragment{1}{k^2}$, $H_i$ and $H_i'$ are perfectly matched by $\mA$.
    \end{claim}
    \begin{claimproof}
        For each $i \in \fragment{1}{k^2}$, let $c_i$ denote the number of characters in $H_i$ that are deleted or substituted in the alignment $\mA$, and define $c_i'$ analogously for $H_i'$. Note that $\sum_i c_i = \sum_i c_i'$, since they are the only substrings that contain the character $\#$.

        Now, consider modifying $\mA$ to $\mA'$ by perfectly matching each $H_i$ to $H_i'$, and deleting any characters that were previously involved in substitutions with characters in $H_i$. By \cref{claim:clique_gadget:2}, we have that $\mA'$ is still a valid alignment.
        
        Let $s$ denote the number of substitutions that are removed in this process. Then the change in cost is $\cost(\mA') - \cost(\mA) = s - \sum_i c_i - \sum_i c_i'$. Since $s \leq \sum_i c_i + \sum_i c_i'$, this change is not positive, so also $\mA'$ must be optimal.
        By our choice of $\mA$, the number of pairs $(i, j)$ such that $B_i[j]$ is matched to $B_i'[j]$ cannot increase. Therefore, if such number remains unchanged then $\sum_i c_i = 0$, which implies that $\mA$ already matches each $H_i$ perfectly with $H_i'$.
    \end{claimproof}
    Now observe that \cref{claim:clique_gadget:3} implies
    \[
        \ed(\CLG(X), \CNG(Y)) = \sum_i \ed(B_i, B_i') = \sum_{v \in Y, u \in X} \ed(\CG(v), \LG(u)).
    \]
    Since $X \subseteq N(Y)$ (or equivalently $Y \subseteq N(X)$) holds if and only if $v \in N(u)$ for all $u \in X$, $v \in Y$, it follows from \cref{claim:clique_gadget:1} that we can set $C \coloneqq C' k^2$ to conclude the proof.
\end{proof}

Next, we present the trees in our lower-bound construction.
As part of our construction, we use the label set $\Sigma$ from \cref{cor:clique_gadged} and introduce four additional labels: $\S$, $\$$, and $\#_\ell$ for $\ell \in \{\ttL, \ttR\}$. We assume these new labels are disjoint from $\Sigma$.
Further, in our lower-bound instance we use the gadgets $\CLG, \CNG$ from \cref{cor:clique_gadged} of lengths $\lambda_1$ and $\lambda_2$, respectively.
We set $\lambda \coloneqq \max(\lambda_1, \lambda_2)$.

\begin{definition}\label{def:instance}
    Let $\bG = (V, E)$ be a graph, let $k \in \mathbb{Z}_{+}$,
    and let $\mX, \mY$, and $\mZ$ all denote the set of $k$-cliques in $\bG$. Suppose $|\mX| = |\mY| = |\mZ| = N$. For $Z \in \mZ$ define the two trees $\bT_\mX(Z)$ and $\bT_\mY'(Z)$ as follows:
    \begin{enumerate}
    \item $\bT_\mX(Z)$ is obtained by attaching the following nodes
    to $\bP(\$ \circ (\ \bigcirc_{X \in \mX} \ \CLG(X) \ ) \circ \S^{16\lambda N + \lambda + 1})$:
    \label{def:instance:1}
    \begin{enumerate}[(a)]
        \item We left attach $\#_\ttL^{5\lambda N}$ to $\$$,
        \label{def:instance:1a}
        \item For each $X \in \mX$, we right attach $\#_\ttR^{4\lambda} \circ \CNG(X)$ to the last node of $\CLG(X)$ on the path, and \label{def:instance:1b}
        \item We left attach $\CNG(Z) \circ \#_\ttL^{5\lambda N}$ 
        to the last node belonging to any $\CLG(X)$ for $X \in \mX$ (so to the spine node right before the first node with label $\S$).
        \label{def:instance:1c}
    \end{enumerate}
    \item $\bT_\mY'(Z)$ is obtained by attaching the following nodes
    to $\bP(\$ \circ (\ \bigcirc_{Y \in \mY} \ \CNG(Y) \ ) \circ \S^{16\lambda N + \lambda + 1})$:
    \label{def:instance:2}
    \begin{enumerate}[(a)]
        \item We right attach $\#_\ttR^{5\lambda N}$ to $\$$, 
        \label{def:instance:2a}
        \item For each $Y \in \mY$, we left attach $\CLG(Y) \circ \#_\ttL^{5\lambda}$ to the last node of $\CNG(Y)$ on the path, and
        \label{def:instance:2b}
        \item We right attach $\#_\ttR^{5\lambda N} \circ \CLG(Z)$ 
        to the last node belonging to any $\CNG(Y)$ for $Y \in \mY$ (so to the spine node right before the first node with label $\S$).
        \label{def:instance:2c} \qedhere
    \end{enumerate}
    \end{enumerate}
\end{definition}

Refer to \cref{fig:unweighted_instance} for a visualization of \( \bT_\mX(Z) \) and \( \bT_\mY'(Z) \).
We remark that in \( \bT_\mX(Z) \) and \( \bT_\mY'(Z) \), each gadget $\mathrm{G}(W)$ for \( \mathrm{G} \in \{\CLG, \CNG\} \) and \( W \in \mX \cup \mY \cup \{Z\} \) appears exactly once in each tree. Therefore, by a slight abuse of notation, we will also use \( \mathrm{G}(W) \) to refer to the set of nodes comprising that gadget.

\begin{figure*}[htbp]
   \centering
   \input{figures/unweighted_instance}
   \caption{
       The trees $\bT_\mX(Z)$ and $\bT_\mY'(Z)$ depicted on the left and right, respectively.
    }
   \label{fig:unweighted_instance}
\end{figure*}
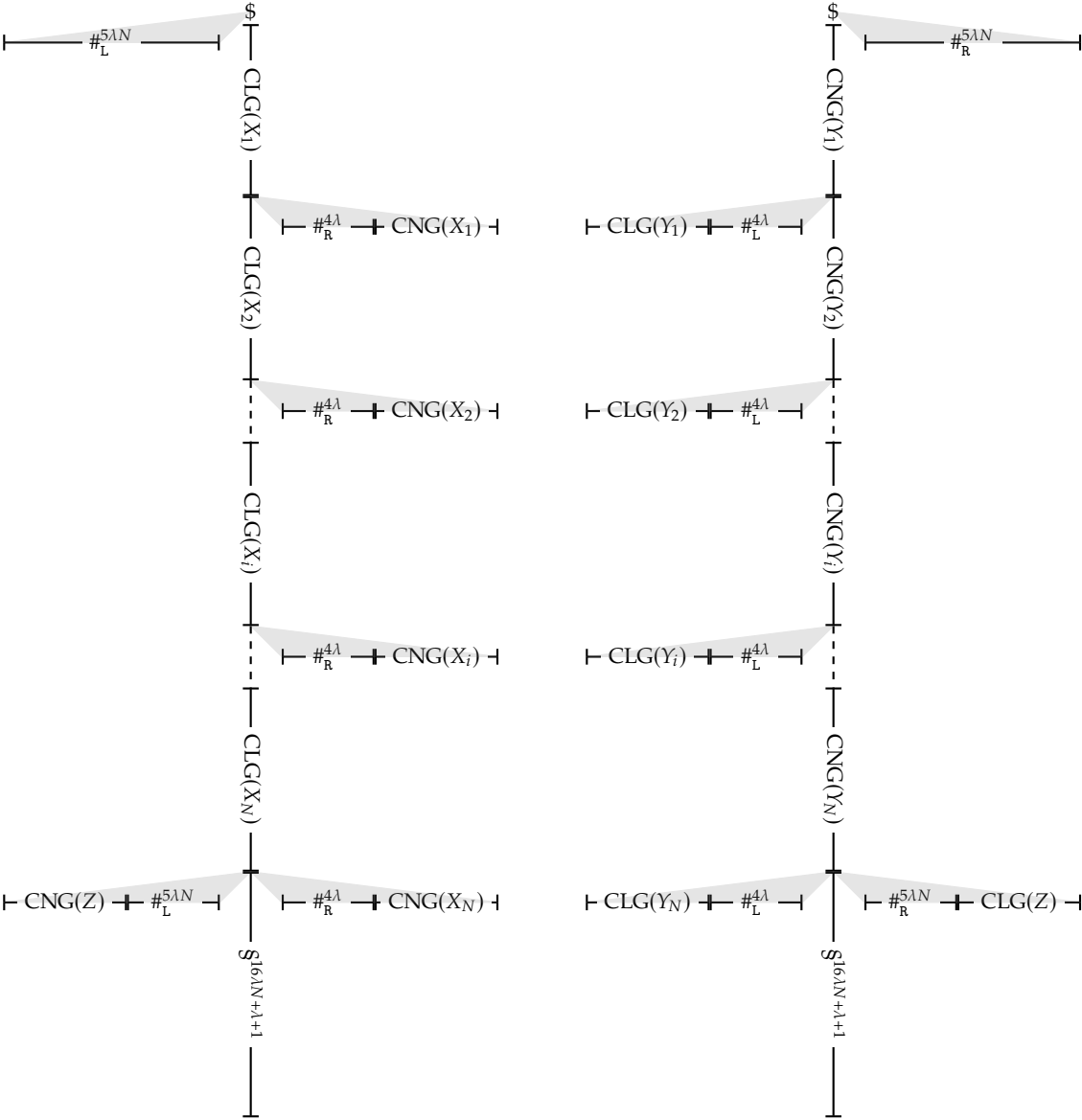

\begin{theorem}\label{thm:unweighted_embedding}
    We have  $\ted(\bT_\mX(Z), \bT_\mY'(Z))$ equals to
    \[
        \min_{Y \in \mY,X \in \mX} \big( \ \ed(\CLG(X), \CNG(Y)) + \ed(\CNG(X), \CLG(Z)) + \ed(\CLG(Y), \CNG(Z)) \ \big) + D,
    \]
    for $D = 12\lambda N - \lambda_2 N - \lambda_1 N - \lambda_2 - \lambda_1$.
\end{theorem}

\begin{proof}
    Throughout this proof, we abbreviate \( \bT_\mX(Z) \) and \( \bT_\mY'(Z) \) as \( \bT \) and \( \bT' \), respectively. Additionally, let \( X_1, \ldots, X_N \) denote the sets in \( \mX \) such that \( \CLG(X_i) \) appears on the spine above \( \CLG(X_{i+1}) \) in \( \bT \), and let $\bL$ (resp. $\bR$) be the nodes on the left (resp. right) of the spine of $\bT$.
    We define \( Y_1, \ldots, Y_N\) and \(\bL',\bR'\) analogously for \( \mY \) and \( \bT' \).

    Central in this proof is the alignment \( \mA_{a,b} \) of \( \bT \) onto \( \bT' \), structured for \( a, b \in \fragment{1}{N} \) as follows:
    \begin{enumerate}[(i)]
             \item The roots and the last $16\lambda N + \lambda + 1$ nodes of the spines of $\bT,\bT'$ are perfectly matched. \label{it:c_unweighted_embedding_1}
            \item The three pairs of gadgets $\CLG(X_a),\CNG(Y_b)$ and $\CNG(X_a),\CLG(Z)$ and $\CLG(Y_b),\CNG(Z)$ are aligned as the corresponding strings would be under string \ED. \label{it:c_unweighted_embedding_2}
            \item Match or substitute the nodes of $\bR \setminus \CNG(X_a)$ with nodes labeled $\#_{\ttR}$ as follows.
            The ones above the gadget $\CLG(X_a)$ with the nodes of \ref{def:instance:2}\ref{def:instance:2a}, while those below with the nodes of \ref{def:instance:2}\ref{def:instance:2c}. Note that $|\bR| = 5\lambda N$, so all nodes of $\bR \setminus \CNG(X_a)$ can be matched or substituted as described, as in \ref{def:instance:2}\ref{def:instance:2a} and \ref{def:instance:2}\ref{def:instance:2c} we add enough to do so.
            \label{it:c_unweighted_embedding_3}
            \item Similarly, match or substitute the nodes of $\bL' \setminus \CLG(Y_b)$ with nodes labeled $\#_{\ttL}$ like so:
            The ones above the gadget $\CLG(Y_b)$ with the nodes of \ref{def:instance:1}\ref{def:instance:1a}, while those below with the nodes of \ref{def:instance:1}\ref{def:instance:1c}.
            \label{it:c_unweighted_embedding_4}
            \item All remaining nodes are deleted.
            \label{it:c_unweighted_embedding_5}
        \end{enumerate}
        
    \begin{claim}
        The alignment $\mA_{a,b}$ has cost 
        \[
           \ed(\CLG(X_a), \CNG(Y_b)) + \ed(\CNG(X_a), \CLG(Z)) + \ed(\CLG(Y_b), \CNG(Z)) + D.
        \]
    \end{claim}
    \begin{claimproof}
        We show that the cost claimed above can be charged to the the construction steps of $\mA_{a,b}$ as:
        \begin{align*}
            &\underbrace{0}_{\text{\ref{it:c_unweighted_embedding_1}}} 
            + \underbrace{\ed(\CLG(X_a), \CNG(Y_b)) + \ed(\CNG(X_a), \CLG(Z)) + \ed(\CLG(Y_b), \CNG(Z))}_{\text{\ref{it:c_unweighted_embedding_2}}} \\
            &+ \underbrace{\lambda_1(N-1) + \lambda_2(N-1)}_{\text{\ref{it:c_unweighted_embedding_3} and \ref{it:c_unweighted_embedding_4}}}
            + \underbrace{\lambda_2(N-1) + \lambda_1(N-1) + 6\lambda N - \lambda_1 N + \lambda_1 + 6\lambda N - \lambda_2 N + \lambda_2}_{\text{\ref{it:c_unweighted_embedding_5}}}.
        \end{align*}
        The cost of steps \ref{it:c_unweighted_embedding_1} and \ref{it:c_unweighted_embedding_2} does not need any proof.
        
        For step \ref{it:c_unweighted_embedding_3}, note that in $\bR \setminus \CNG(X_a)$ there are exactly $\lambda_1 (N - 1)$ nodes without the label $\#_\ttR$ (that is, all $\CLG(X_{a'})$ for $a \neq a'$). These nodes are substituted and incur a cost of $\lambda_1 (N - 1)$. The other nodes in $\bR \setminus \CNG(X_a)$ are matched exactly, and incur no cost.
        Following a similar reasoning, in \ref{it:c_unweighted_embedding_4} we incur cost $\lambda_2 (N - 1)$.
    
        The deleted nodes of \ref{it:c_unweighted_embedding_5} amount to:
        \begin{itemize}
            \item All $\CLG(X_{a'}), \CNG(Y_{b'})$ for $a \neq a'$ and $b \neq b'$ get deleted, incurring $\lambda_2(N-1) + \lambda_1(N-1)$ in cost.
            \item Out of the overall $10\lambda N$ nodes labeled $\#_{\ttR}$ from $\bR'$, only $4\lambda N + \lambda_1(N-1)$ are matched or substituted in \ref{it:c_unweighted_embedding_3}, so the remaining $10\lambda N - (4\lambda N + \lambda_1(N-1)) = 6\lambda N - \lambda_1 N + \lambda$ get deleted, incurring $6\lambda N - \lambda_1 N + \lambda_1$ in cost.
            \item Similar argument holds for the nodes labeled $\#_{\ttL}$ from $\bL'$, incurring $6\lambda N - \lambda_2 N + \lambda_2$ in cost. \claimqedhere
        \end{itemize}
    \end{claimproof}
        
    In order to prove \cref{thm:unweighted_embedding}, it suffices to show that any optimal alignment \( \mA \) transforming \( \bT \) into \( \bT' \) has cost at least as high as that of some alignment \( \mA_{a,b} \). To do so, we will demonstrate that \( \mA \) satisfies all construction steps, \ref{it:c_unweighted_embedding_1} through \ref{it:c_unweighted_embedding_5}, for some choice of \( a, b \in \fragment{1}{N} \).

    \begin{claim}\label{clm:unweighted_embedding:1} 
        For an optimal alignment $\mA$, the nodes on the spine of $\bT$ must be matched or substituted by $\mA$ with nodes on the spine of $\bT'$, and similarly, the nodes of $\bR$ (resp. $\bL$) must be matched or substituted by $\mA$ with the nodes of $\bR'$ (resp. $\bL'$).
    \end{claim}
    \begin{claimproof}
        It suffices to argue that $\mA$ matches at least a pair of nodes labeled $\S$ with each other, then by the rules of tree alignments, the structure specified in the claim follows. To this end, suppose for the sake of contradiction that every node labeled $\S$ is deleted or substituted, each incurring a cost of at least 1 and resulting in a total cost of at least $2(16\lambda N + \lambda + 1)$. By matching the roots and all nodes labeled $\S$ with each other and deleting the remaining nodes, we incur cost $4 \cdot 5\lambda N + \lambda_2 + \lambda_1 + 2\cdot 4 \lambda N + 2 \lambda_2 N + 2\lambda_1 N < 2(16\lambda N + \lambda + 1)$, which contradicts the optimality of $\mA$.
    \end{claimproof}

    By \Cref{clm:unweighted_embedding:1}, we may assume that $\mA$ satisfy the three following:
    \begin{enumerate}[(a)]
        \item Let $v,v' \in \bT$ be two spine nodes such that $v$ is an ancestor of $v'$, and such that $v$ and $v'$ but none of $(\sub(v) \setminus \{v\}) \setminus \sub(v')$, are matched or substituted by $\mA$. Then, none of the spine nodes between $v$ and $v'$ share the same label as $v$.
        \item Let $v,v'\in \bL$ be two nodes such that $v$ occurs in the pre-order of $\bT$ before $v'$,
        and such that $v$ and $v'$ but none of the nodes in the pre-order between $v$ and $v'$, are matched or substituted by $\mA$. Then, none of the nodes in the pre-order between $v$ and $v'$ share the same label as $v$.
        \item Similarly, let $v,v' \in \bR$ be nodes such that $v$ occurs in the post-order of $\bT$ after $v'$,
        and such that $v$ and $v'$ but none of the nodes in the post-order between $v$ and $v'$, are matched or substituted by $\mA$. Then, none of the nodes in the post-order between $v$ and $v'$ share the same label as $v$. 
    \end{enumerate}

    Whenever $\mA$ has this form, we say $\mA$ \emph{favors matching downwards}.
    Note, if $\mA$ does not satisfy this structural form, then it can be modified (iteratively) so that it does, without changing its cost.
    
    \begin{claim}\label{clm:unweighted_embedding:2} 
        The roots labeled $\$$, as well as the last $16\lambda N + \lambda + 1$ spine nodes labeled $\S$, in $\bT$ and $\bT'$ are matched with each other by $\mA$.\
        That is, construction step \ref{it:c_unweighted_embedding_1} is satisfied.
    \end{claim}
    \begin{claimproof}
        Clearly, the optimal alignment must match the two roots with each other.

        Suppose, for contradiction, that for some \( i \in \fragment{1}{16\lambda N + \lambda + 1} \), the \( i \)-th last node on the spine of \( \bT \) and the $i$-th last node on the spine of \( \bT' \) are not matched with each other. Let \( i \) be the smallest such index. In the situation where either (1) both are deleted; or (2) exactly one is substituted with some other node whose label is not $\S$, we can further reduce the cost of $\mA$ by matching these two nodes together, contradicting the optimality of \( \mA \).
        Thus, at most one of the two nodes is substituted with another node labeled \( \S \), but this contradicts our earlier assumption that \( \mA \) favors matching downwards.
    \end{claimproof}

    Next, consider the largest $i \in \fragment{1}{N}$ such that a node right-attached to $\CLG(X_{i-1})$ is matched with a node labeled \( \#_{\ttR} \) from \ref{def:instance:2}\ref{def:instance:2a} (if $i = 1$ then such a node does not exist).
    Moreover, consider the smallest $i' \in \fragment{1}{N}$ such that a node right-attached to $\CLG(X_{i'})$ is matched with a node labeled \( \#_{\ttR} \) from \ref{def:instance:2}\ref{def:instance:2c}. Note that $i'$ is well-defined.
    Indeed, in $\mA$ the right-attached nodes to $\CLG(X_N)$ with label $\#_{\ttR}$ must match with nodes from \ref{def:instance:2}\ref{def:instance:2c} because $\mA$ favors matching downwards. Symmetrically, we define $j,j'$ indexing $Y_1, \ldots, Y_N$.

    In the remainder of the proof, we show that \( \mA \) satisfies \ref{it:c_unweighted_embedding_2}, \ref{it:c_unweighted_embedding_3}, \ref{it:c_unweighted_embedding_4}, and \ref{it:c_unweighted_embedding_5}, for \( a = i \) and \( b = j \).
    \begin{claim}\label{clm:unweighted_embedding:3} 
        The indices $i,i'$ satisfy \( i' \geq i \).
    \end{claim}
    \begin{claimproof}
        First, note that since \( \mA \) is a valid tree alignment, we must have \( i' \geq i - 1 \). To eliminate the remaining case \( i' = i - 1 \), for the sake of contradiction, suppose \( i' = i - 1 \), which we can only have if \( i > 1 \).

        Observe that the condition \( i > 1 \) implies that among the \( 4\lambda N + \lambda_1 N \) nodes in \( \bR \), there exists a node \( v \in \bR \) that is not matched with a node labeled \( \#_{\ttR} \) from \ref{def:instance:2}\ref{def:instance:2c}, but instead is matched with a node \( w \in \bR' \) from \ref{def:instance:2}\ref{def:instance:2a}. Let \( w \) be the first such matched node encountered in the pre-order traversal of \( \bT' \).

        As a result, there must also exist a node \( u \in \bR' \) labeled \( \#_{\ttR} \) from \ref{def:instance:2}\ref{def:instance:2c} that is neither matched nor substituted. By the assumption that \( \mA \) favors downward matches, we may assume \( u \) is the last node visited in the pre-order traversal of \( \bT' \).
        
        Finally, observe that the assumption \( i' = i - 1 \) implies that \( \CLG(Z) \) must be deleted. Since no nodes between \( w \) and \( u \) are matched or substituted, this contradicts the assumption that \( \mA \) favors matching downward.
    \end{claimproof}

    \begin{claim}\label{clm:unweighted_embedding:4} 
        The indices $i,i'$ satisfy \( i = i' \).
    \end{claim}
    \begin{claimproof}
        Suppose, for contradiction, that \( i' > i \). Then all of the \( 4\lambda \) nodes labeled \( \#_{\ttR} \), which are right-attached to \( \CLG(X_{i'-1}) \), are deleted. This implies that also at least \( 4\lambda \) nodes labeled \( \#_{\ttR} \) from \ref{def:instance:2}\ref{def:instance:2c} must be deleted.
        By the assumption that \( \mA \) favors downward matching, we may assume these deletions occur among the last \( 4\lambda \) nodes visited in the pre-order traversal of \( \bT' \).

        We now modify \( \mA \) as follows: we match these \( 4\lambda \) last-visited nodes with the \( 4\lambda \) right-attached \( \#_{\ttR} \)-labeled nodes of \( \CLG(X_{i'-1}) \), and instead delete all nodes from \( \CLG(X_{i'-1}) \), \( \CNG(X_{i'-1}) \), as well as any nodes previously matched to the deleted gadgets (at most \( 2\lambda_2 + 2\lambda_1 \) nodes in total). This modification reduces the overall cost of $\mA$ by at least $8\lambda - 2\lambda_2 + 2\lambda_1 \geq 4\lambda > 0$, contradicting its optimality.
    \end{claimproof}

    It is now not difficult to see, by \cref{clm:unweighted_embedding:1} and \cref{clm:unweighted_embedding:4}, that all nodes in \( \bR \setminus \CNG(X_i) \) are matched to nodes labeled \( \#_{\ttR} \) from \ref{def:instance:2}\ref{def:instance:2a} and \ref{def:instance:2}\ref{def:instance:2c}, respectively, as specified in step \ref{it:c_unweighted_embedding_3}. By symmetry, we have \( j = j' \), and construction step \ref{it:c_unweighted_embedding_4} is satisfied as well.
    Moreover, by the optimality of \( \mA \), the three pairs of gadgets \( \CLG(X_i), \CNG(Y_j) \), \( \CNG(X_i), \CLG(Z) \), and \( \CLG(Y_j), \CNG(Z) \) must be aligned in the same way their corresponding strings would be under standard string $\ED$, thus satisfying step \ref{it:c_unweighted_embedding_2}. All remaining nodes must be deleted, satisfying step \ref{it:c_unweighted_embedding_5}.
\end{proof}

This allows us to conclude this section with the dynamic lower bound.

\unweightedted

\begin{proof}
    Let $k$ be a sufficiently large constant to be determined later. We show that an algorithm for unweighted \DTED, as described in the statement, can be used to construct a combinatorial algorithm for detecting a $3k$-clique with running time $\Oh(n^{3k - \varepsilon'})$ for some $\varepsilon' > 0$.

    Given an instance $\bG$ of $3k$-Clique Detection, the algorithm proceeds as follows. We first discover the set of all $k$-cliques in $\bG$, which we denote with $\mS$.
    The algorithm performs $\Oh(n^k)$ rounds, one for each $k$-clique $Z \in \mS$. In round $Z$, we dynamically maintain the two trees $\bT_{\mX}(Z)$ and $\bT_{\mY}'(Z)$ where $\mX = \mY = \mS$ as defined in \cref{def:instance}.
    Note that between any two consecutive rounds $Z$ and $Z'$, the corresponding trees $\bT_{\mX}(Z), \bT_{\mX}(Z)$ and $\bT_{\mY}'(Z), \bT_{\mY}'(Z')$  differ only in the gadgets $\CNG(Z)$, and $\CLG(Z)$, respectively. 
    Both gadgets have length $\lambda = \Oh(nk \log n)$, and can be transformed into one another using $\Oh(\lambda)$ updates. After constructing the instance for each round $Z$, we query the distance $d_Z \coloneqq \ted(\bT_{\mX}(Z), \bT_{\mY}(Z))$ and keep track of the minimum value $d \coloneqq \min_Z d_Z$ across all rounds. If $d - D \leq 3C$ for $C,D$ as defined in \cref{cor:clique_gadged} and \cref{thm:unweighted_embedding}, respectively, then we return \yes. Otherwise, we return \no.
    
    By \cref{thm:unweighted_embedding}, we have $d_Z - D \leq 3C$ if and only if there exist $k$-cliques $X, Y \in \mS$ such that $X \cup Y \cup Z$ forms a $3k$-clique. Since we iterate over all $Z \in \mS$, the algorithm checks for the existence of a $3k$-clique in $\bG$.

    The total size of the unweighted instance maintained throughout is $N = \Oh(n^{k+1} \log n)$, treating $k$ as a constant. This is also the number of updates and queries performed.
    Therefore, if the total running time can be bounded by $\Oh(p(N) + N \cdot (u(N) + q(N))) = \Oh(N^{3(1 - \epsilon)})$, then for a large enough $k$ and small enough $\epsilon' > 0$, this evaluates to
    $\Oh(n^{3(k+1)(1 - \epsilon)} \log^{3(1 - \epsilon)} n) = \Oh(n^{3k - \varepsilon'})$, contradicting \cref{conj:cliquedetection}.
\end{proof}

%% file: figures/unweighted_instance.tex
\begin{tikzpicture}[]

\begin{scope}[scale=0.85, every node/.append style={transform shape}]
    \node at (0, 0) {$\$$};
    \draw[thick, |-|] (0, -0.2) -- node[fill=white, rotate=-90] {$\CLG(X_1)$} (0, -3);
    \draw[thick, |-|] (0, -3) -- node[fill=white, rotate=-90] {$\CLG(X_2)$} (0, -6);
    \draw[thick, dashed] (0, -6) -- (0, -7);
    \draw[thick, |-|] (0, -7) -- node[fill=white, rotate=-90] {$\CLG(X_{i})$} (0, -10);
    \draw[thick, dashed] (0, -10) -- (0, -11);
    \draw[thick, |-|] (0, -11) -- node[fill=white, rotate=-90] {$\CLG(X_{N})$} (0, -14);
    \draw[thick, |-|] (0, -14) -- node[fill=white, rotate=-90] {$\S^{16\lambda N + \lambda + 1}$} (0, -18);

    \draw[thick, |-|] (-4,-0.5) -- node[fill=white] {$\#_{\ttL}^{5\lambda N}$} (-0.5,-0.5);
    \draw[fill, gray, fill opacity=0.2, opacity=0.2] (0,0) -- (-4,-0.5) -- (-0.5,-0.5) -- (0,0);
    
    \draw[thick, |-|] (0.5, -3.5) -- node[fill=white] {$\#_{\ttR}^{4\lambda}$} (2, -3.5);
    \draw[thick, |-|] (2,-3.5) -- node[fill=white] {$\CNG(X_1)$} (4,-3.5);
    \draw[fill, gray, fill opacity=0.2, opacity=0.2] (0, -3) -- (0.5, -3.5) -- (4,-3.5) -- (0, -3);
    
    \draw[thick, |-|] (0.5, -6.5) -- node[fill=white] {$\#_{\ttR}^{4\lambda}$} (2, -6.5);
    \draw[thick, |-|] (2,-6.5) -- node[fill=white] {$\CNG(X_2)$} (4,-6.5);
    \draw[fill, gray, fill opacity=0.2, opacity=0.2] (0, -6) -- (0.5, -6.5) -- (4,-6.5) -- (0, -6);
    
    \draw[thick, |-|] (0.5, -10.5) -- node[fill=white] {$\#_{\ttR}^{4\lambda}$} (2, -10.5);
    \draw[thick, |-|] (2,-10.5) -- node[fill=white] {$\CNG(X_i)$} (4,-10.5);
    \draw[fill, gray, fill opacity=0.2, opacity=0.2] (0, -10) -- (0.5, -10.5) -- (4,-10.5) -- (0, -10);
    
    \draw[thick, |-|] (0.5, -14.5) -- node[fill=white] {$\#_{\ttR}^{4\lambda}$} (2, -14.5);
    \draw[thick, |-|] (2,-14.5) -- node[fill=white] {$\CNG(X_N)$} (4,-14.5);
    \draw[fill, gray, fill opacity=0.2, opacity=0.2] (0, -14) -- (0.5, -14.5) -- (4,-14.5) -- (0, -14);
    
    \draw[thick, |-|] (-4,-14.5) -- node[fill=white] {$\CNG(Z)$} (-2,-14.5);
    \draw[thick, |-|] (-2,-14.5) -- node[fill=white] {$\#_{\ttL}^{5\lambda N}$} (-0.5,-14.5);
    \draw[fill, gray, fill opacity=0.2, opacity=0.2] (0,-14) -- (-4,-14.5) -- (-0.5,-14.5) -- (0,-14);
\end{scope}

\begin{scope}[shift={(8, 0)}, scale=0.85, every node/.append style={transform shape}]
    \node at (0, 0) {$\$$};
    \draw[thick, |-|] (0, -0.2) -- node[fill=white, rotate=-90] {$\CNG(Y_1)$} (0, -3);
    \draw[thick, |-|] (0, -3) -- node[fill=white, rotate=-90] {$\CNG(Y_2)$} (0, -6);
    \draw[thick, dashed] (0, -6) -- (0, -7);
    \draw[thick, |-|] (0, -7) -- node[fill=white, rotate=-90] {$\CNG(Y_{i})$} (0, -10);
    \draw[thick, dashed] (0, -10) -- (0, -11);
    \draw[thick, |-|] (0, -11) -- node[fill=white, rotate=-90] {$\CNG(Y_{N})$} (0, -14);
    \draw[thick, |-|] (0, -14) -- node[fill=white, rotate=-90] {$\S^{16\lambda N + \lambda + 1}$} (0, -18);

    \draw[thick, |-|] (4,-0.5) -- node[fill=white] {$\#_{\ttR}^{5\lambda N}$} (0.5,-0.5);
    \draw[fill, gray, fill opacity=0.2, opacity=0.2] (0,0) -- (4,-0.5) -- (0.5,-0.5) -- (0,0);
    
    \draw[thick, |-|] (-0.5, -3.5) -- node[fill=white] {$\#_{\ttL}^{4\lambda}$} (-2, -3.5);
    \draw[thick, |-|] (-2,-3.5) -- node[fill=white] {$\CLG(Y_1)$} (-4,-3.5);
    \draw[fill, gray, fill opacity=0.2, opacity=0.2] (0, -3) -- (-0.5, -3.5) -- (-4,-3.5) -- (0, -3);
    
    \draw[thick, |-|] (-0.5, -6.5) -- node[fill=white] {$\#_{\ttL}^{4\lambda}$} (-2, -6.5);
    \draw[thick, |-|] (-2,-6.5) -- node[fill=white] {$\CLG(Y_2)$} (-4,-6.5);
    \draw[fill, gray, fill opacity=0.2, opacity=0.2] (0, -6) -- (-0.5, -6.5) -- (-4,-6.5) -- (0, -6);
    
    \draw[thick, |-|] (-0.5, -10.5) -- node[fill=white] {$\#_{\ttL}^{4\lambda}$} (-2, -10.5);
    \draw[thick, |-|] (-2,-10.5) -- node[fill=white] {$\CLG(Y_i)$} (-4,-10.5);
    \draw[fill, gray, fill opacity=0.2, opacity=0.2] (0, -10) -- (-0.5, -10.5) -- (-4,-10.5) -- (0, -10);
    
    \draw[thick, |-|] (-0.5, -14.5) -- node[fill=white] {$\#_{\ttL}^{4\lambda}$} (-2, -14.5);
    \draw[thick, |-|] (-2,-14.5) -- node[fill=white] {$\CLG(Y_N)$} (-4,-14.5);
    \draw[fill, gray, fill opacity=0.2, opacity=0.2] (0, -14) -- (-0.5, -14.5) -- (-4,-14.5) -- (0, -14);
    
    \draw[thick, |-|] (4,-14.5) -- node[fill=white] {$\CLG(Z)$} (2,-14.5);
    \draw[thick, |-|] (2,-14.5) -- node[fill=white] {$\#_{\ttR}^{5\lambda N}$} (0.5,-14.5);
    \draw[fill, gray, fill opacity=0.2, opacity=0.2] (0,-14) -- (4,-14.5) -- (0.5,-14.5) -- (0,-14);
\end{scope}

\end{tikzpicture}

%% file: s5_weighted.tex
\section{Lower Bounds for Dynamic Weighted \TED}
\label{sec:weighted}

Our lower bound construction builds upon Bringmann et al. \cite{BGMW20}, where they reduce Min-Weight 3-Clique to \TED. Their construction of the \TED instance is in the style of the tree-alignment formulation as in \cref{def:tree-alignment}. Moreover, they do the following tweaks on the cost function: Denoting the old cost function with $\delta_{old}$, they define 
\begin{align*}
    \delta_{new}(\ell(v), \varepsilon) &\coloneqq 0\\
    \delta_{new}(\varepsilon, \ell(v)) &\coloneqq 0\\
    \delta_{new}(\ell(v),\ell(v')) &\coloneqq \delta_{old}(\ell(v),\ell(v')) - \delta_{old}(\ell(v),\varepsilon) - \delta_{old}(\varepsilon, \ell(v'))
\end{align*}
for all $v, v'$. Note that this definition does not change the optimal alignment. The advantage that the new $\delta$ offers is that now we can ignore nodes that are not aligned by $\mA$ (deleting them has 0 cost), and $\cost(\mA)$ consists solely of the aligned pairs. 

Given a Min-Weight 3-Clique instance $\bG = (V,E,w)$, to make the weight function complete, we additionally define $w(i,j) = \infty$  whenever $(i,j) \notin E$, and $w(i,i) = \infty$ for all node $i$.

\begin{definition}[\TED Instance \cite{BGMW20}]\label{def:karl_instance}
Given an instance $\bG = (V, E, w)$ for Min-Weight-3-Clique with $w$ augmented as above, we define an instance $(\bT_1, \bT_2, \Sigma, \delta)$ for \TED as follows.
\begin{itemize}
    \item $\Sigma = \{a_1, \dots, a_n, b_1, \dots, b_{n+1}, a_1', \dots, a_n', b_1', \dots, b_{n+1}', c_1, \dots, c_n, d_1, \dots, d_{n+1}, c_1', \dots, c_n', d_1', \dots, d_{n+1}'\}$.
    \item $\bT_1$ is obtained from right-attaching the following nodes to $\bP(a_1 \cdots a_n b_1 \cdots b_{n+1})$: For every $i \in [n]$, right-attach a leaf node $a_i'$ to $a_i$, and right-attach $b_i'$ to $b_i$.
    \item $\bT_2$ is obtained from left-attaching the following nodes to $\bP(c_1 \cdots c_n d_1 \cdots d_{n+1})$: For every $i \in [n]$, left-attach a leaf node $c_i'$ to $c_i$, and left-attach $d'_i$ to $d_i$.
    \item The cost $\delta(u, v)$ of matching nodes $u$ and $v$ is defined as follows, with $M$ being a sufficiently large constant:
    \begin{enumerate}
    \item $\delta(b'_k, d'_k) = -M^2-2M \cdot k$ for every $k \in \fragment{1}{n}$
    \label{it:karl_instance:1}
    \item $\delta(b_{k+1}, c'_j) = -M^2+M \cdot k + M \cdot j + w(k,j)$ for every $k \in \fragment{1}{n}, j \in \fragment{1}{n}$
    \label{it:karl_instance:2}    
    \item $\delta(a'_i, d_{k+1}) = -M^2 + M \cdot k + M \cdot i + w(i,k)$ for every $i \in \fragment{1}{n}, k \in \fragment{1}{n}$
    \label{it:karl_instance:3}
    \item $\delta(a_i, c_j) = -2M + w(i,j) - w(i-1, j-1)$ for every $i \in \fragment{2}{n}, j \in \fragment{2}{n}$
    \label{it:karl_instance:4}
    \item $\delta(a_i, c_1) = -M(i+1) + w(i,1)$ for every $i \in \fragment{1}{n}$
    \label{it:karl_instance:5}
    \item $\delta(a_1, c_j) = -M(j+1) + w(1,j)$ for every $i \in \fragment{1}{n}$
    \label{it:karl_instance:6}
    \item All other costs are set to $\infty$. \qedhere
    \label{it:karl_instance:7}
\end{enumerate}
\end{itemize}
\end{definition}

Hence, we can view each of $\bT_1$ and $\bT_2$ as having a ``top'' part (the $a$-nodes and $c$-nodes), and a ``bottom'' part (the $b$-nodes and $d$-nodes).  Note that this instance has $\Oh(n)$ nodes and $\Oh(n)$ alphabet size. Bringmann et al. show that the minimum weight of a triangle can be extracted from the cost of an optimal tree alignment.
\begin{definition}\label{def:Aijk}
    Given a \TED instance $(\bT_1, \bT_2, \Sigma, \delta)$ as in \cref{def:karl_instance}, consider a tree-alignment consisting of the following matchings, where $i, j, k \in \fragment{1}{n}$.
    \begin{enumerate}[(a)]
    \item A matching $(b_k', d_k')$ 
    \item A matching $(b_{k+1}, c_j')$
    \item A matching $(a'_i, d_{k+1})$ 
    \item Consecutive matchings $(a_i,c_j), (a_{i-1}, c_{j-1}), \dots, (a_{i-j+2}, c_2)$ if $i \geqslant j$; otherwise $(a_i,c_j), (a_{i-1}, c_{j-1}), \dots, (a_{2}, c_{j-i+2})$ 
    \item A matching $(a_{i-j+1},c_1)$ if $i \geqslant j$; otherwise $(a_{1},c_{j-i+1})$ 
    \end{enumerate}

We call a tree-alignment with the above structure $\mA_{i,j,k}$. Some simple calculations show that $\mA_{i,j,k}$ has cost $-3M^2 + w(i,k) + w(k,j)+ w(i,j)$.
\end{definition}

\begin{theorem} [\cite{BGMW20}] \label{thm:ted_align_cost}
Given a \TED instance $(\bT_1, \bT_2, \Sigma, \delta)$ as in \cref{def:karl_instance}, any optimal tree-alignment is of the form $\mA_{i,j,k}$ for some $i,j,k \in \fragment{1}{n}$. Therefore, an optimal tree-alignment selects the $i,j,k$ to achieve the minimum cost 
\[
-3M^2+\min_{i,j,k} w(i,k) + w(k,j)+ w(i,j)
\] And the minimum weight of a triangle can be extracted from the alignment cost. Moreover, each tree alignment not having the form $\mA_{i,j,k}$ has cost at least $-3M^2+M$, for large enough $M$. \lipicsEnd
\end{theorem}

Let us give some intuition on why these matchings encode a triangle. Note that weight \eqref{it:karl_instance:1}, \eqref{it:karl_instance:2}, and \eqref{it:karl_instance:3} have a ``$-M^2$'' summand in them, so in order to minimize the total matching cost, the optimal alignment will for sure choose a matching of the form $(b'_{k_1}, d'_{k_1})$, a matching of the form $(b_{k_2+1}, c'_j)$, and a matching of the form $(a'_i, d_{k_3+1})$ for some $k_1, k_2, k_3, i, j$. Moreover, there can only be \emph{one} matching of each type in order to not violate the rules of tree alignment. In fact, one can show that $k_1 = k_2 =k_3$ achieves the lowest cost. We can think of matchings $(k_2,j)$ and $(k_3,i)$ as representing edges $(k,j)$ and $(k,i)$ for some $k$ in the original graph. Finally, to bridge $(i,j)$, the spine-to-spine matchings form a telescoping sum, and the costs eventually cancels out to be $w(i,j)$. 

Now, let us introduce our \DTED instance. See Figure~\ref{fig:caterpillar_no_matchings} for an illustration.
\begin{definition}[\DTED Instance] \label{def:dynamic_TED}
    Given an instance $\bG = (V, E, w)$ for Min-Weight 4-Clique, and a node $x \in \fragment{1}{n}$,
    build a $\TED$ instance $(\bT_1^x, \bT_2^x, \Sigma^x, \delta^x)$ from $(\bT_1, \bT_2, \Sigma, \delta)$
    from \cref{def:karl_instance} as follows:
    \begin{itemize}
        \item $\Sigma^x = \Sigma \cup \{a_1'', \ldots, a_n'', c_1'', \ldots, c_n'', b_1'', \ldots, b_n'', \hat{b}_{x}, \hat{d}_{x}, \hat{c}_{x}, \bot\}$.
        \item For every $i \in \fragment{1}{n}$, attach $a_i''$ as the child of $a_i'$, $c_i''$ as the children of $c_i'$, and $b_i''$ as the first child of $b_i$ (hence the left sibling of $b_i'$).
        \item Attach $\hat{b}_{x}$ as the child of $b_{n+1}$, $\hat{d}_{x}$ as the child of $d_{n+1}$, and add node $\bot$ between $c_n$ and $d_1$. Attach $\hat{c}_{x}$ as a left child of $\bot$.
        \item We define $\delta^x(\cdot, \cdot)$ out of $\delta$. That is, we keep costs \eqref{it:karl_instance:1}, \eqref{it:karl_instance:2}, \eqref{it:karl_instance:3}, \eqref{it:karl_instance:4}, \eqref{it:karl_instance:5}, and \eqref{it:karl_instance:6} from \cref{def:karl_instance}, and additionally define the following costs:
        \begin{enumerate}
        \setcounter{enumi}{6}
        \item $\delta(\hat{b}_{x}, c''_j) = -M + w(x,j)$ for every $j \in \fragment{1}{n}$.
        \item $\delta(\hat{d}_{x}, a''_i) = -M + w(x,i)$ for every $i \in \fragment{1}{n}$.
        \item $\delta(\hat{c}_{x}, b''_k) = -M + w(x,k)$ for every $k \in \fragment{1}{n}$.
        \item All other costs are set to $\infty$. \qedhere
        \end{enumerate}
    \end{itemize}
\end{definition}
    
Note that the alphabet size and number of nodes of our dynamic instance is still $\Oh(n)$. We call the part of the instance from \cref{def:karl_instance} the \emph{underlying tree}.

\begin{lemma}\label{lem:dted_align_cost}
    For a large enough $M$, an optimal tree-alignment in the above \DTED instance has the structure of $\mA_{i,j,k}$ on the underlying trees from Definition~\ref{def:karl_instance}, and three additional matchings $(\hat{b}_{x}, c_j''), (\hat{d}_{x}, a_i''), (\hat{c}_{x}, b_k'')$. Thus, an optimal alignment will choose $i, j, k$ that gives the minimum cost 
    \[
-3M^2-3M+\min_{i,j,k}{w(i,j) + w(i,k) + w(j,k) + w(x,i) + w(x,j) + w(x,k)}
\]
\end{lemma}
\begin{proof}
Let $\OPT$ denote the optimal alignment cost. First, note that the above alignment structure is a valid tree alignment. Thus $\OPT \leqslant -3M^2-3M+\min_{i,j,k}{w(i,j) + w(i,k) + w(j,k) + w(x,i) + w(x,j) + w(x,k)}$. 

It only remains to show that every alignment has at least this cost. For $M$ large enough, in an optimal alignment, either we choose all three matchings $(b_{k_1}', d_{k_1}')$, $(b_{k_2+1}, c_j')$, and $(a_i', d_{k_3+1})$, or the total matching cost up to lower order terms is at least $-2M^2$. We can argue that $k_1=k_2=k_3=k$ for some $k$. This is because if $k_1 < k_2$, then we can either increase $k_1$ or decrease $k_2$ to further lower the cost, and likewise for $k_1$ and $k_3$. Now that we have fixed these three matchings, an optimal alignment will also match $\hat{b}_{x}$ with $c_j''$, $\hat{d}_{x}$ with $a_i''$, and $\hat{c}_{x}$ with $b_k''$, since adding these matchings decreases the total cost by another $3M$, up to lower order terms. Note that by the definition of a tree alignment, these are the \emph{only} possible matchings for $\hat{b}_{x}, \hat{d}_{x}$, and $\hat{c}_{x}$. Now, by our definition of $\delta$, the only possible additional alignments are the spine-to-spine matchings. Note that our matchings can be partitioned into two groups: the matchings containing $\hat{d}_x, \hat{c}_x, \hat{d}_x$ and the remaining matchings of the underlying trees from \cref{def:karl_instance}. If the spine-to-spine matchings do not have the form $(a_i, c_j),\dots ,(a_{i-j+1}, c_1)$, then by \cref{thm:ted_align_cost}, we have that the tree-alignment on the underlying tree has cost at least $-3M^2+M$. Hence the tree-alignment on the $\DTED$ instance has cost at least $-3M^2-2M$, which is not optimal for $M$ large enough.
\end{proof}

This allows us to conclude this (sub)section with the dynamic lower bound.

\weightedted

\begin{proof}
    Given a Min-Weight 4-Clique instance $\bG = (V, E, \w)$ with $n$ nodes, we construct our algorithm to perform $n$ rounds, one for each node $x \in V$, while maintaining the \DTED instance $(\bT_1^x, \bT_2^x, \Sigma^x, \delta^x)$ throughout the rounds. Note that between consecutive rounds $x$ and $y$, we only change the label of $\hat{b}_x, \hat{c}_x, \hat{d}_x$ into $\hat{b}_y, \hat{c}_y, \hat{d}_y$. 
    In each round $x$, we query the tree edit distance between $\bT_1^x$ and $\bT_2^x$, and keep track of the minimum value across all rounds (adding an offset of $3M^3 + 3M$ to each). 
    
    By \cref{lem:dted_align_cost}, the result of round $x$ allows us to identify nodes $i, j, k$ such that ${i, j, k, x}$ form a min-weight 4-clique containing $x$. After we iterate over all $x \in V$, we indeed return the min-weight 4-clique. The size and alphabet size of our \DTED instance is $\Oh(n)$, and there are $\Oh(n)$ updates and queries. Thus if there exists an algorithm for \DTED running in overall running time $\Oh(p(N) + N(u(N) + q(N))) = \Oh(N^{4-\epsilon})$, we would solve Min-Weight 4-Clique in time $\Oh(n^{4-\epsilon})$. So for a Min-Weight 4-Clique instance $\bG = (V, E, \w)$ satisfying the weight condition as in \cref{conj:weightedclique}, it would contradict \cref{conj:weightedclique}.
\end{proof}

\subparagraph*{Lower Bounds for Incremental/Decremental Dynamic \TED.}

An easy modification of our construction allows us to also establish conditional lower bounds for incremental (resp., decremental) \DTED, where the updates are restricted to node insertions (resp., deletions). Let us briefly sketch the reduction to incremental \DTED here. The lower bound for decremental \DTED can be shown similarly.

First we assign some order to the $n$ nodes $x_1, \dots, x_n$ in the Min-Weight 4-Clique instance. Then we define $\delta_{incr}$ to contain all costs in $\delta$ from \cref{def:karl_instance}, and additionally define 
\begin{enumerate}
\setcounter{enumi}{6}
\item $\delta_{incr}(\hat{b}_{x_\ell}, c_j'') = -M\cdot \ell + w(x_\ell,j)$ for every $\ell,j \in \fragment{1}{n}$
\item $\delta_{incr}(\hat{d}_{x_\ell}, a_i'') = -M\cdot \ell + w(x_\ell,i)$ for every $\ell,i \in \fragment{1}{n}$
\item $\delta_{incr}(\hat{c}_{x_\ell}, b_k'') = -M\cdot \ell + w(x_\ell,k)$ for every $\ell,k \in \fragment{1}{n}$
\item All other costs are set to $\infty$
\end{enumerate}
Starting from an empty graph, we make $\Oh(n)$ insertions to arrive at the \DTED instance $\bT_1^{x_1}, \bT_2^{x_2}$ as in \cref{def:dynamic_TED}. We make a query to our black-box incremental \DTED algorithm to get a min-weight 4-clique containing $x_1$. This concludes the first round. For round $\ell$, $2 \leq \ell \leq n$, we attach $\hat{b}_{x_\ell}$ as the child of $\hat{b}_{x_{\ell-1}}$, $\hat{c}_{x_\ell}$ as the child of $\hat{c}_{x_{\ell-1}}$, and $\hat{d}_{x_\ell}$ as the child of $\hat{d}_{x_{\ell-1}}$, and query for a min-weight 4-clique containing $x_\ell$. Taking the minimum weight over all $n$ rounds gives us a desired min-weight 4-clique. 

To see the correctness, note that once $\hat{b}_{x_\ell}, \hat{c}_{x_\ell}, \hat{d}_{x_\ell}$ are inserted, they gain priority over all previously inserted $\hat{b}, \hat{c}, \hat{d}$-nodes to be matched. So, essentially, we are considering each node $x \in V$ independently as in the previous reduction to \DTED. As there are $\Oh(n)$ updates and queries in this reduction, we can conclude that unless \cref{conj:weightedclique} fails, for any $\varepsilon > 0$, there is no algorithm for incremental \DTED of size $\Oh(N)$ and alphabet size $\Oh(N)$ that satisfies $p(N) + N (u(N) + q(N)) = \Oh(N^{4-\varepsilon})$.

\begin{figure*}[htbp]
     \centering
     \begin{subfigure}[b]{1\textwidth}
         \centering
         \scalebox{0.7}{\input{figures/weighted_instance_1}}
         \caption{The \TED Instance from \cite{BGMW20}. The numeric labels indicate the ranges of a certain type of matchings.}
         \label{fig:weighted:a}
     \end{subfigure}
     \newline
     \hfill
     
     \begin{subfigure}[b]{1\textwidth}
         \centering
         \scalebox{0.7}{\input{figures/weighted_instance_2}}
         \caption{Extension of the \TED instance.}
         \label{fig:weighted:b}
     \end{subfigure}
     
    \caption{
        The figure illustrates how to extend the \TED instance from \cite{BGMW20} to obtain a dynamic bound. Nodes newly introduced in \cref{fig:weighted:b} are shown in full color, while those from \cref{fig:weighted:a} are rendered translucent.
     }
     \label{fig:caterpillar_no_matchings}
     \label{fig:with_matchings}
\end{figure*}

%% file: figures/weighted_instance_1.tex
\begin{tikzpicture}[x=0.75pt,y=0.75pt,yscale=-1,xscale=1]
%uncomment if require: \path (0,461); %set diagram left start at 0, and has height of 461

%%%%%%%%%%%%% Tree 1
\draw[dotted, rounded corners] (-8,18) rectangle node[left=10pt] {\eqref{it:karl_instance:4}, \eqref{it:karl_instance:5} \eqref{it:karl_instance:6}} (8,214);
\draw[dotted, rounded corners] (-8,317) rectangle node[left=10pt] {\eqref{it:karl_instance:2}} (8,512);
\draw[dotted, rounded corners] (38,317) rectangle node[below right=10pt] {\eqref{it:karl_instance:1}} (54,512);
\draw[dotted, rounded corners] (15,42) rectangle node[above right=10pt]{\eqref{it:karl_instance:3}} (31,236);

%Straight Line
\draw    (0,24) -- (0,506) ;
% 1st layer of top nodes
\foreach \y in {24, 47, 70, 93, 116, 139, 162, 185, 208} {
        \draw[fill=black] (0, \y) circle [radius=3.5];
    }

\node at (-15, 24) {$a_1$};
\node at (-15, 47) {$a_2$};
\node at (-15, 208) {$a_n$};

% 2nd layer of top nodes
\foreach \y in {47, 70, 93, 116, 139, 162, 185, 208, 231} {
        \draw[fill=black] (23, \y) circle [radius=3.5];
    }
\node at (42, 50) {$a'_1$};
\node at (42, 76) {$a'_2$};
\node at (42, 237) {$a'_n$};

% % 3rd layer of top nodes
% \foreach \y in {70, 93, 116, 139, 162, 185, 208, 231, 254} {
%         \draw[fill=black] (46, \y) circle [radius=3.5];
%     }
% \node at (46, 57) {$a''_1$};

% connect the layers
\draw (0,24) -- (23, 47);
\draw (0, 47) -- (23, 70);
\draw (0, 70) -- (23, 93);
\draw (0, 93) -- (23, 116);
\draw (0, 116) -- (23, 139);
\draw (0, 139) -- (23, 162);
\draw (0, 162) -- (23, 185);
\draw (0, 185) -- (23, 208);
\draw (0, 208) -- (23, 231);

% bottom nodes
% 1st layer nodes
\foreach \y in {300, 323, 346, 369, 392, 415, 438, 461, 484, 506} {
        \draw[fill=black] (0, \y) circle [radius=3.5];
    }
% x-node
% \draw (0, 506) -- (-23, 529);
% \draw[fill=orange, draw=orange] (-23, 529) circle [radius=4];

\node at (-15, 300) {$b_1$};
\node at (-15, 323) {$b_2$};

\node at (62, 326) {$b_1'$};
\node at (62, 353) {$b_2'$};
\node at (62, 510) {$b_n'$};

% \node at (-23, 542) {$e_{x_1}$};

% 2nd layer of bottom nodes

% \foreach \y in {323, 346, 369, 392, 415, 438, 461, 484, 506} {
%         \draw[fill=black] (23, \y) circle [radius=3.5];
%     }
% \node at (23, 519) {$b''_n$};

% 3rd layer of bottom nodes
\foreach \y in {323, 346, 369, 392, 415, 438, 461, 484, 506} {
        \draw[fill=black] (46, \y) circle [radius=3.5];
    }
% \node at (46, 519) {$b'_n$};

% connext bottom nodes
\draw (0, 300)  -- (46, 323);
\draw (0, 323)  -- (46, 346);
\draw (0, 346)  -- (46, 369);
\draw (0, 369)  -- (46, 392);
\draw (0, 392)  -- (46, 415);
\draw (0, 415)  -- (46, 438);
\draw (0, 438)  -- (46, 461);
\draw (0, 461)  -- (46, 484);
\draw (0, 484)  -- (46, 506);

% \draw (0, 300)  -- (23, 323);
% \draw (0, 323)  -- (23, 346);
% \draw (0, 346)  -- (23, 369);
% \draw (0, 369)  -- (23, 392);
% \draw (0, 392)  -- (23, 415);
% \draw (0, 415)  -- (23, 438);
% \draw (0, 438)  -- (23, 461);
% \draw (0, 461)  -- (23, 484);
% \draw (0, 484)  -- (23, 506);

%%%%%%%%%%%% Tree 2
%Straight Lines 
\draw[dotted, rounded corners] (392,20) rectangle node[right=10pt]{\eqref{it:karl_instance:4}, \eqref{it:karl_instance:5} \eqref{it:karl_instance:6}} (408,212);
\draw[dotted, rounded corners] (369,42) rectangle node[below left=10pt] {\eqref{it:karl_instance:2}} (385,236);
\draw[dotted, rounded corners] (369,317) rectangle node[above left=10pt]{\eqref{it:karl_instance:1}} (385,512);
\draw[dotted, rounded corners] (392,317) rectangle node[right=10pt]{\eqref{it:karl_instance:3}} (408,512);

\draw    (400,24) -- (400,506) ;
% 1st layer of top nodes
\foreach \y in {24, 47, 70, 93, 116, 139, 162, 185, 208} {
        \draw[fill=black] (400, \y) circle [radius=3.5];
    }
\node at (415, 24) {$c_1$};
\node at (415, 47) {$c_2$};
\node at (415, 208) {$c_n$};

% 2nd layer of top nodes
\foreach \y in {47, 70, 93, 116, 139, 162, 185, 208, 231} {
        \draw[fill=black] (377, \y) circle [radius=3.5];
    }
\node at (357, 50) {$c'_1$};
\node at (357, 73) {$c'_2$};
\node at (357, 234) {$c'_n$};

% 3rd layer top nodes
% \foreach \y in {70, 93, 116, 139, 162, 185, 208, 231, 254} {
%         \draw[fill=black] (354, \y) circle [radius=3.5];
%     }
% \node at (354, 57) {$c''_1$};
% connect top layer nodes]
\draw (400, 24) -- (377, 47);
\draw (400, 47) -- (377, 70);
\draw (400, 70) -- (377, 93);
\draw (400, 93) -- (377, 116);
\draw (400, 116) -- (377, 139);
\draw (400, 139) -- (377, 162);
\draw (400, 162) -- (377, 185);
\draw (400, 185) -- (377, 208);
\draw (400, 208) -- (377, 231);

% \draw (400, 260) -- (377, 283);
% \draw[fill=white] (400, 260) circle [radius=3.5];
% \draw[fill=orange, draw=orange] (377, 283) circle [radius=4];
% \node at (377, 296) {$e_{x_3}$};

% 1st layer bottom nodes
\foreach \y in {300, 323, 346, 369, 392, 415, 438, 461, 484, 506} {
        \draw[fill=black] (400, \y) circle [radius=3.5];
    }
\node at (415, 300) {$d_1$};
\node at (415, 323) {$d_2$};
\node at (419, 506) {$d_{n+1}$};

\node at (357, 326) {$d_1'$};
\node at (357, 349) {$d_2'$};
\node at (357, 510) {$d_n'$};

% 2nd layer bottom nodes
\foreach \y in {323, 346, 369, 392, 415, 438, 461, 484, 506} {
        \draw[fill=black] (377, \y) circle [radius=3.5];
    }

% \draw[fill=orange, draw=orange] (423, 529) circle [radius=4];
% \draw (400, 506) -- (423, 529);

% connect bottom nodes
\draw (400, 300) -- (377, 323);
\draw (400, 323) -- (377, 346);
\draw (400, 346) -- (377, 369);
\draw (400, 369) -- (377, 392);
\draw (400, 392) -- (377, 415);
\draw (400, 415) -- (377, 438);
\draw (400, 438) -- (377, 461);
\draw (400, 461) -- (377, 484);
\draw (400, 484) -- (377, 506);

% matchings

\draw[fill=red, draw=red] (46, 461) circle [radius=3.5];
\draw[fill=red, draw=red] (377, 461) circle [radius=3.5];
\draw[-, red, ultra thick] (46, 461) to (377, 461);

\draw[fill=red, draw=red] (0, 461) circle [radius=3.5];
\draw[fill=red, draw=red] (377, 93) circle [radius=3.5];
\draw[-, red, ultra thick] (0, 461) to[bend left = 10] (377, 93);

\draw[fill=red, draw=red] (400, 461) circle [radius=3.5];
\draw[fill=red, draw=red] (23, 208) circle [radius=3.5];
\draw[-, red, ultra thick] (23, 208) to[bend right = 10] (400, 461);

\draw[fill=red, draw=red] (400, 24) circle [radius=3.5];
\draw[fill=red, draw=red] (400, 47) circle [radius=3.5];
\draw[fill=red, draw=red] (400, 70) circle [radius=3.5];

\draw[fill=red, draw=red] (0, 139) circle [radius=3.5];
\draw[fill=red, draw=red] (0, 162) circle [radius=3.5];
\draw[fill=red, draw=red] (0, 185) circle [radius=3.5];

\draw[-, red, ultra thick] (400, 24) -- (0, 139);
\draw[-, red, ultra thick] (400, 47) -- (0, 162);
\draw[-, red, ultra thick] (400, 70) -- (0, 185);

\node at (-15, 506) {$b_{n+1}$};

\end{tikzpicture}

%% file: figures/weighted_instance_2.tex
\begin{tikzpicture}[x=0.75pt,y=0.75pt,yscale=-1,xscale=1]

\def\opac{0.3};

\begin{scope}[color=gray, opacity=\opac]

    \foreach \y in {24, 47, 70, 93, 116, 139, 162, 185, 208} {
            \node[fill=gray, circle, minimum size=7pt, inner sep=0pt] (a\y) at (0, \y) {};
        }
    \foreach \y in {47, 70, 93, 116, 139, 162, 185, 208, 231} {
            \node[fill=gray, circle, minimum size=7pt, inner sep=0pt] (ap\y) at (23, \y) {};
        }

    % connect the layers
    \draw (a24) -- (ap47);
    \draw (a47) -- (ap70);
    \draw (a70) -- (ap93);
    \draw (a93) -- (ap116);
    \draw (a116) -- (ap139);
    \draw (a139) -- (ap162);
    \draw (a162) -- (ap185);
    \draw (a185) -- (ap208);
    \draw (a208) -- (ap231);

    % bottom nodes
    % 1st layer nodes
    \foreach \y in {300, 323, 346, 369, 392, 415, 438, 461, 484, 506} {
            \node[fill=gray, circle, minimum size=7pt, inner sep=0pt] (b\y) at (0, \y) {};
        }
  
    % 2nd layer of bottom nodes
    
    % 3rd layer of bottom nodes
    \foreach \y in {323, 346, 369, 392, 415, 438, 461, 484, 506} {
            \node[fill=gray, circle, minimum size=7pt, inner sep=0pt] (bp\y) at (46, \y) {};
        }
    % \node at (46, 519) {$b'_n$};

    \draw (a24) -- (a47) -- (a70) -- (a93) -- (a116) -- (a139) -- (a162) -- (a185) -- (a208) -- (b300) -- (b323) -- (b346) -- (b369) -- (b392) -- (b415) -- (b438) -- (b461) -- (b484) -- (b506);
    
    % connext bottom nodes
    \draw (b300)  -- (bp323);
    \draw (b323)  -- (bp346);
    \draw (b346)  -- (bp369);
    \draw (b369)  -- (bp392);
    \draw (b392)  -- (bp415);
    \draw (b415)  -- (bp438);
    \draw (b438)  -- (bp461);
    \draw (b461)  -- (bp484);
    \draw (b484)  -- (bp506);

    % 1st layer of top nodes
    \foreach \y in {24, 47, 70, 93, 116, 139, 162, 185, 208} {
            \node[fill=gray, circle, minimum size=7pt, inner sep=0pt] (c\y) at (400, \y) {};
        }
    
    % 2nd layer of top nodes
    \foreach \y in {47, 70, 93, 116, 139, 162, 185, 208, 231} {
            \node[fill=gray, circle, minimum size=7pt, inner sep=0pt] (cp\y) at (377, \y) {};
        }

    \draw (c24) -- (cp47);
    \draw (c47) -- (cp70);
    \draw (c70) -- (cp93);
    \draw (c93) -- (cp116);
    \draw (c116) -- (cp139);
    \draw (c139) -- (cp162);
    \draw (c162) -- (cp185);
    \draw (c185) -- (cp208);
    \draw (c208) -- (cp231);

    % 1st layer bottom nodes
    \foreach \y in {300, 323, 346, 369, 392, 415, 438, 461, 484, 506} {
            \node[fill=gray, circle, minimum size=7pt, inner sep=0pt] (d\y) at (400, \y) {};
        }
    % 2nd layer bottom nodes
    \foreach \y in {323, 346, 369, 392, 415, 438, 461, 484, 506} {
            \node[fill=gray, circle, minimum size=7pt, inner sep=0pt] (dp\y) at (377, \y) {};
        }

    % connect bottom nodes
    \draw (d300)  -- (dp323);
    \draw (d323)  -- (dp346);
    \draw (d346)  -- (dp369);
    \draw (d369)  -- (dp392);
    \draw (d392)  -- (dp415);
    \draw (d415)  -- (dp438);
    \draw (d438)  -- (dp461);
    \draw (d461)  -- (dp484);
    \draw (d484)  -- (dp506);

     \draw (c24) -- (c47) -- (c70) -- (c93) -- (c116) -- (c139) -- (c162) -- (c185) -- (c208) -- (d300) -- (d323) -- (d346) -- (d369) -- (d392) -- (d415) -- (d438) -- (d461) -- (d484) -- (d506);
    
    % matchings

    \node[fill=red, circle, minimum size=7pt, inner sep=0pt] (r1) at (46, 461) {};
    \node[fill=red, circle, minimum size=7pt, inner sep=0pt] (r2) at (377, 461) {};
    \draw[-, red, ultra thick, opacity=\opac] (r1) to (r2);

    \node[fill=red, circle, minimum size=7pt, inner sep=0pt] (r3) at (0, 461) {};
    \node[fill=red, circle, minimum size=7pt, inner sep=0pt] (r4) at (377, 93) {};
    \draw[-, red, ultra thick, opacity=\opac] (r3) to[bend left = 10] (r4);

    \node[fill=red, circle, minimum size=7pt, inner sep=0pt] (r5) at (400, 461) {};
    \node[fill=red, circle, minimum size=7pt, inner sep=0pt] (r6) at (23, 208) {};
    \draw[-, red, ultra thick, opacity=\opac] (r5) to[bend left = 10] (r6);

    \node[fill=red, circle, minimum size=7pt, inner sep=0pt] (s1) at (400, 24) {};
    \node[fill=red, circle, minimum size=7pt, inner sep=0pt] (s2) at (400, 47) {};
    \node[fill=red, circle, minimum size=7pt, inner sep=0pt] (s3) at (400, 70) {};
    
    \node[fill=red, circle, minimum size=7pt, inner sep=0pt] (s4) at (0, 139) {};
    \node[fill=red, circle, minimum size=7pt, inner sep=0pt] (s5) at (0, 162) {};
    \node[fill=red, circle, minimum size=7pt, inner sep=0pt] (s6) at (0, 185) {};

    \draw[-, red, ultra thick, opacity=\opac] (s1) -- (s4);
    \draw[-, red, ultra thick, opacity=\opac] (s2) -- (s5);
    \draw[-, red, ultra thick, opacity=\opac] (s3) -- (s6);
\end{scope}

\draw (d506) -- (423, 529);
\draw[fill=red, draw=red] (423, 529) circle [radius=4] node[below] {$\hat{d}_x$};

% 3rd layer top nodes
\foreach \y in {70, 93, 116, 139, 162, 185, 208, 231, 254} {
        \draw[fill=black] (354, \y) circle [radius=3.5];
    }

\draw (354, 70) node[left=5pt] {$c_1''$} -- (cp47);
\draw (354, 93) node[left=5pt] {$c_2''$} -- (cp70);
\draw (354, 116) -- (cp93);
\draw (354, 139) -- (cp116);
\draw (354, 162) -- (cp139);
\draw (354, 185) -- (cp162);
\draw (354, 208) -- (cp185);
\draw (354, 231) -- (cp208);
\draw (354, 254) node[left=5pt] {$c_n''$} -- (cp231);

% \node at (354, 57) {$c''_1$};

\draw (b300)  -- (23, 323) node[below right] {$b_1''$};
\draw (b323)  -- (23, 346) node[below right] {$b_2''$};
\draw (b346)  -- (23, 369);
\draw (b369)  -- (23, 392);
\draw (b392)  -- (23, 415);
\draw (b415)  -- (23, 438);
\draw (b438)  -- (23, 461);
\draw (b461)  -- (23, 484);
\draw (b484)  -- (23, 506) node[below right] {$b_n''$};

\draw (b506) -- (-23, 529);
\draw[fill=red, draw=red] (-23, 529) circle [radius=4];
\node at (-23, 542) {$\hat{b}_x$};

2nd layer of bottom nodes

\foreach \y in {323, 346, 369, 392, 415, 438, 461, 484, 506} {
        \draw[fill=black] (23, \y) circle [radius=3.5];
    }

 \draw[black] (400, 260) -- (377, 283);
\draw[fill=white, black] (400, 260) circle [radius=3.5] node[above left] {$\bot$};
\draw[fill=red, draw=red] (377, 283) circle [radius=4];
\node[black] at (377, 296) {$\hat{c}_x$};

% 3rd layer of top nodes
\foreach \y in {70, 93, 116, 139, 162, 185, 208, 231, 254} {
        \draw[fill=black] (46, \y) circle [radius=3.5];
    }

\draw (46, 70) node[right=5pt] {$a_1''$} -- (ap47);
\draw (46, 93) node[right=5pt] {$a_2''$} -- (ap70);
\draw (46, 116) -- (ap93);
\draw (46, 139) -- (ap116);
\draw (46, 162) -- (ap139);
\draw (46, 185) -- (ap162);
\draw (46, 208) -- (ap185);
\draw (46, 231) -- (ap208);
\draw (46, 254) node[right=5pt] {$a_n''$} -- (ap231);
    
% \node at (46, 57) {$a''_1$};

\draw[fill=red, draw=red] (23, 461) circle [radius=3.5];
\draw[-, red, ultra thick] (23, 461) to[bend left = 10] (377, 283);

\draw[fill=red, draw=red] (46, 231) circle [radius=3.5];
\draw[-, red, ultra thick] (46, 231) to[bend right = 40] (423, 529);

\draw[fill=red, draw=red] (354, 116) circle [radius=3.5];
\draw[-, red, ultra thick] (-23, 529) to[bend right = 29] (354, 116);

\end{tikzpicture}

%% file: s6_rnafolding.tex
\section{Lower Bounds for Dynamic RNA Folding and Dyck Edit Distance}
\label{sec:rnafolding}

 We build upon the static lower bound instance from \cite{ABVW15}, summarized as follows.

\begin{lemma}[\cite{ABVW15}]\label{lem:rna_static_string}
Let $\bG = (V, E)$ be a graph on $n = |V|$ nodes, and let $k \in \mathbb{Z}_{+}$. Then, there exist embeddings
$\CG_{\ttA},\CG_{\ttB},\CG_{\ttC} : V^k \rightarrow {\{\Sigma \cup \Sigma'\}}^{\leq \ell}$,
with $\ell = \Oh(n)$ over a constant-size alphabet $\Sigma$ and a weight function $\w : \Sigma \rightarrow [M]$ with $M = \Oh(k^4 n \log n)$, such that for sets $\mX, \mY, \mZ \subseteq V^k$ of size $N$, the string
\begin{align*}
    S \coloneqq \ & \#_{\ttA}^{2N} \ \big(\bigcirc_{X \in \mX} \ \#_{\ttA}' \ \CG_{\ttA}(X) \ \#_{\ttA}'\  \big) \ \#_{\ttA}^{2N} \circ \\
    & \#_{\ttB}^{2N} \big(\ \bigcirc_{Y \in \mY} \ \#_{\ttB}'\  \CG_{\ttB}(Y) \ \#_{\ttB}' \ \big) \ \#_{\ttB}^{2N} \circ \\
    & \#_{\ttC}^{2N} \big(\bigcirc_{Z \in \mZ} \ \#_{\ttC}' \ \CG_{\ttC}(Z)  \ \#_{\ttC}' \ \big) \ \#_{\ttC}^{2N}
\end{align*}
satisfies the two following.
\begin{enumerate}[(i)]
    \item There exists a positive constant $C$ such that $\score_{\w}(S) = C$ if there are $X \in \mX$, $Y \in \mY$, and $Z \in \mZ$ such that $X \cup Y \cup Z$ forms a $3k$-clique in $\bG$. Otherwise, $\score_{\w}(S) < C$. \label{it:rna_static_string:i}
    \item For every optimal folding \( F \) of \( S \), there exists a set \( X \in \mX \) such that if \( \CG_{\ttA}(X) \) appears within the substring \( S\fragmentco{x}{x+\lambda} \), then for all pairs \( (i, j) \in F \) with \( \{S\position{i}, S\position{j}\} = \{\#_{\ttA}, \#_{\ttA}'\} \), it holds that either \( i < j < x \) or \( x + \lambda \leq i < j \). 
    Similarly, there exist sets \( Y \in \mY \) and \( Z \in \mZ \) such that the same condition holds with respect to the delimiters \( \#_{\ttB}, \#_{\ttB}' \) and \( \#_{\ttC}, \#_{\ttC}' \), respectively. If $\score_{\w}(S) = C$, 
    then this conditions holds for all $X,Y,Z$ forming a $3k$-clique.
    \label{it:rna_static_string:ii}
    \lipicsEnd
\end{enumerate}
\end{lemma}

We proceed to modify this instance to obtain a dynamic lower bound.

\begin{lemma}\label{lem:rna_dynamic_string}
    Consider the string $S$ and the weight function $\w$ from \cref{lem:rna_static_string}:
    \begin{itemize}
        \item For $W \in V^k$, modify the string $S$ to get the string
        \begin{align*}
            S_W \coloneqq \ & \CLG_{\ttA}(W) \ \#_{\ttA}^{2N} \ \big(\bigcirc_{X \in \mX} \  \#_{\ttA}' \ p(\CNG_{\ttA}(X))^{R} \ \CG_{\ttA}(X) \ \#_{\ttA}'\  \big) \ \#_{\ttA}^{2N} \circ \\
            & \CLG_{\ttB}(W) \ \#_{\ttB}^{2N} \big(\ \bigcirc_{Y \in \mY} \  \#_{\ttB}'\   p(\CNG_{\ttB}(Y))^{R}\ \CG_{\ttB}(Y) \ \#_{\ttB}' \ \big) \ \#_{\ttB}^{2N} \circ \\
            & \CLG_{\ttC}(W) \ \#_{\ttC}^{2N} \big(\bigcirc_{Z \in \mZ} \ \#_{\ttC}' \  p(\CNG_{\ttC}(Z))^{R} \ \CG_{\ttC}(Z)  \ \#_{\ttC}' \ \big) \ \#_{\ttC}^{2N},
        \end{align*}
        where \( \CLG_{\ttA} \) and \( \CNG_{\ttA} \) are the strings \( \CLG \) and \( \CNG \) from \cref{lem:clique_gadged}, defined over a new alphabet \( \Sigma_{\ttA} \) disjoint from \( \Sigma \).  
        Similarly, \( \CLG_{\ttB}, \CNG_{\ttB} \) and \( \CLG_{\ttC}, \CNG_{\ttC} \) are defined over new alphabets \( \Sigma_{\ttB} \) and \( \Sigma_{\ttC} \), respectively, each disjoint from one another and from \( \Sigma \cup \Sigma_{\ttA} \).
        \item Modify the weight function $\w$ from \cref{lem:rna_static_string} to get $\w' : \Sigma \cup \Sigma_{\ttA} \cup \Sigma_{\ttB} \cup \Sigma_{\ttC} \rightarrow \fragment{1}{4\lambda M}$
        defined as
        \[
            \w'(\sigma) = 
            \begin{cases}
                4\lambda \w(\sigma) & \sigma \in \Sigma\\
                1 & \sigma \in \Sigma_{\ttA} \cup \Sigma_{\ttB} \cup \Sigma_{\ttC}
            \end{cases}
        \]
        where $\lambda$ is the length of $\CLG, \CNG$.
    \end{itemize}
    Then, there is a constant $D$ such that $\score_{\w'}(S_W) = D$
    if there are $X \in \mX, Y \in \mY, Z \in \mZ$ such that $X \cup Y \cup Z \cup W$ is a $4k$-clique, and $\score_{\w'}(S_W) < D$, otherwise.
\end{lemma}

\begin{proof} 
    We first prove \cref{claim:rna_dynamic_string:1},
    and then we argue how \cref{lem:rna_dynamic_string} follows from it.

    \begin{claim} \label{claim:rna_dynamic_string:1}
        There are $X \in \mX, Y \in \mY, Z \in \mZ$
        for which the following holds:
        \begin{equation}
            \score_{\w'}(S_W) = 4\lambda \cdot \score_{\w}(S) + \textstyle \sum_{(t, T) \in \{(\ttA, X), (\ttB, Y), (\ttC, Z)\}} \score(\CLG_{t}(W) \circ p(\CNG_{t}(T))^R). \label{eq:rna_dynamic_string}
        \end{equation}     
        Moreover, if $\score_{\w}(S) = C$ (where $C$ is the constant from \cref{lem:rna_static_string}), the sets $X,Y,Z$ are such that $X \cup Y \cup Z$ is a $3k$-clique and such to maximize the summation of scores in  \eqref{eq:rna_dynamic_string}.
    \end{claim}
    \begin{claimproof}
        Let \( F_W \) be an optimal folding of \( S_W \) with score $s_W = \sum_{(i,j) \in F_W} \w(S_W[i])$, and let \( s = \score_{\w}(S) \).  
        Note that \( s > 0 \) if \( N > 0 \).  
        Since \( S \) is a subsequence of \( S_W \), we have \( s_W \geq 4\lambda s \).
        Moreover, as the alphabet \( \Sigma \) is disjoint from \( \Sigma_{\ttA}, \Sigma_{\ttB} \), and \( \Sigma_{\ttC} \), it follows that for any pair \( (i, j) \in F_W \), either both positions \( i \) and \( j \) lie in \( S \), or neither does.

        We can therefore define a folding \( F \) of \( S \) by restricting \( F_W \) to those pairs whose endpoints are both in \( S \). We claim that this folding \( F \) has score \( s \), and must therefore be optimal.  
        Assume for contradiction that this is not the case. Then the score \( s_W \) could be at most   $4\lambda(s - 1) + 3\lambda,$
        as the additional parts of \( S_W \) can contribute at most \( 3\lambda \) to the total score.  
        This implies $s_W \leq 4\lambda(s - 1) + 3\lambda = 4\lambda s - \lambda $, 
        contradicting the earlier bound \( s_W \geq 4\lambda s \).

        Since \( F \) is optimal, it satisfies the property stated in \cref{lem:rna_static_string}\ref{it:rna_static_string:ii}. Importantly, this property also holds for \( F_W \): there is \( X \in \mX \) such that if \( \CG_{\ttA}(X) \) appears within the substring \( S_W\fragmentco{x}{x+\lambda} \), then for any pair \( (i, j) \in F_W \) with \( \{S_W\position{i}, S_W\position{j}\} = \{\#_{\ttA}, \#_{\ttA}'\} \), we have either \( i < j < x \) or \( x + \lambda \leq i < j \). The same holds for some \( Y \in \mY \) and \( Z \in \mZ \). 
        
        At this point, it is not difficult to see that any optimal \( F \) would match the remaining characters of \( S_W \) (those not included in \( S \)) by including the matchings between the three strings \( \CLG_{\ttA}(W) \circ p(\CNG_{\ttA}(X))^R \),
        and \( \CLG_{\ttB}(W) \circ p(\CNG_{\ttB}(Y))^R \), and \( \CLG_{\ttC}(W) \circ p(\CNG_{\ttC}(Z))^R \).

        This completes the first part of the statement.  
        For the second part, observe that if $\score_{\w}(S) = C$, then \cref{lem:rna_static_string}\ref{it:rna_static_string:ii} holds for all triples \( X, Y, Z \) that form a \( 3k \)-clique. Therefore, such triples must maximize the second part of the sum \eqref{eq:rna_dynamic_string}, otherwise, we would contradict the optimality of \( F_W \).  
    \end{claimproof}

    Next, set $D = 4\lambda C + 3C'$, where $C, C'$ are the constants from \cref{lem:rna_static_string} and \cref{lem:clique_gadged}, respectively. 

    Now, suppose there exist \( X \in \mX \), \( Y \in \mY \), and \( Z \in \mZ \) such that \( X \cup Y \cup Z \cup W \) forms a \( 4k \)-clique. Since \( X \cup Y \cup Z \) already forms a \( 3k \)-clique, \cref{claim:rna_dynamic_string:1} implies that the first part of the sum \eqref{eq:rna_dynamic_string} contributes \( 4\lambda C \). Additionally, by \cref{lem:clique_gadged}, the second part must contribute \( 3C' \). Therefore, we have \( \score_{\w'}(S_W) = D \).

    On the other hand, suppose \( \score_{\w'}(S_W) = D \).  
    Then we must have \( \score_{\w}(S) = C \); otherwise, we would get  
    $\score_{\w'}(S_W) \leq 4\lambda(C - 1) + 3\lambda = 4\lambda C - \lambda < D$,  
    which is a contradiction.  
    This implies that the sets \( X \cup Y \cup Z \) form a \( 3k \)-clique.  
    Consequently, the second part of the sum in \eqref{eq:rna_dynamic_string} contributes \( D - 4\lambda C = 3C' \), and by \cref{lem:clique_gadged}, this shows that \( X \cup Y \cup Z \cup W \) forms a \( 4k \)-clique.
\end{proof}

This allows us to conclude this section with the dynamic lower bound.

\rnafolding

\begin{proof}
    Let $k$ be a sufficiently large constant, to be determined later. We show that algorithms for RNA Folding and Dyck Edit Distance, as described in the statement, can be used to construct a combinatorial algorithm for detecting a $4k$-clique with running time $\Oh(n^{4k - \varepsilon'})$ for some $\varepsilon' > 0$. We begin with the case of RNA Folding.

    Given an instance $\bG$ of $4k$-Clique Detection, the algorithm proceeds as follows. First, we discover the set of all $k$-cliques, which we denote with $\mS$. The algorithm performs $n^k$ rounds, one for each $k$-clique $W \in \mS$. In round $W$, we dynamically maintain the RNA Folding instance $S_W$ as defined in \cref{lem:rna_dynamic_string}, where $\mX = \mY = \mZ = \mS$. (Technically, we maintain the unweighted version of the instance given in \cref{lem:weighted_rna_folding}, in which all characters have weight one.)
    Note that between any two consecutive rounds $W$ and $W'$, the corresponding strings $S_W$ and $S_{W'}$ differ only in the gadgets $\CLG_{\ttA}(W)$, $\CLG_{\ttB}(W)$, and $\CLG_{\ttC}(W)$, each of length $\Oh(nk \log n)$. These changes can be handled in $\Oh(nk \log n)$ updates. After constructing the instance for each round $W$, we query its score and keep track of the maximum value across all rounds. 

    By \cref{lem:rna_static_string}, the result of round $W$ reveals whether there exist subsets $X, Y, Z \in \mS$ such that $X \cup Y \cup Z \cup W$ forms a $4k$-clique. Since we iterate over all $W \in \mS$, the algorithm checks for the existence of a $4k$-clique in $\bG$.

    The total size of the unweighted instance maintained throughout is $N = \Oh(n^{k+2} \log^2 n)$, treating $k$ as a constant. This is also the number of updates and queries performed.
    Therefore, the total running time is bounded by $\Oh(p(N) + N \cdot (u(N) + q(N))) = \Oh(N^{4(1 - \epsilon)})$.
    For a sufficiently large $k$ and small enough $\epsilon' > 0$, this evaluates to
    $\Oh(n^{4(k+2)(1 - \epsilon)} \log^{8(1 - \epsilon)} n) = \Oh(n^{4k - \varepsilon'})$.

    For Dyck Edit Distance, it suffices to apply \cref{lem:dyck_to_rnafolding} to the RNA Folding instance maintained above. This transformation increases the size of the instance by at most a constant factor of 6. Moreover, each update to the original string corresponds to at most 6 updates in the transformed Dyck Edit Distance instance.
\end{proof}

%% file: appendix.tex
\section{Online RNA folding and Dyck Edit Distance in OMv Time}
\label{sec:appendix}

The starting point for \cref{thm:online} is a recent paper \cite{DG24}, where Dudek and Gawrychowski consider the online version of the CFG Parsing Problem.

\defproblem
{\textsf{CFG Parsing}}
{A context-free grammar $G$, and a string $S$. }
{\yes if $G$ accepts $S$, and \no otherwise.}

In the online setting, the input is revealed one character at a time, and the task is to compute the answer for each current prefix. Building on Valiant’s result that CFG Parsing can be solved in the same time as Boolean matrix multiplication \cite{Val75}, the authors of \cite{DG24} show that the online CFG Parsing Problem can be solved in \OMv time.

Interestingly, CFG Parsing is closely related to both Dyck Edit Distance and RNA Folding. To make this connection more precise, we now introduce two closely related problems.

\defproblem
{\textsf{Language Edit Distance}}
{A context-free grammar $G$, and a string $S$}
{$\min_{\text{$X$ s.t. $G$ generates $X$}} \ed(X, S)$.}

\defproblem
{\textsf{Scored Parsing}}
{A context-free grammar $G$ where each production rule has an associated cost, and string $S$. }
{The minimum total cost that generates $S$ (if there is any).}

Now, the connection is the following: in \cite{BGSV19}, Bringmann et al. show that Language Edit Distance and RNA Folding can be reduced to the Scored Parsing Problem on bounded-difference grammars\footnote{A CFG $G$ is a bounded-difference grammar if, for every nonterminal $X$, the minimum cost of generating a string $S$ from $X$ changes only slightly when appending a terminal to the end of $S$ or inserting one at the beginning; see \cite{BGSV19} for the formal definition.}. They further reduce this problem, under the assumption of constant grammar size, to a constant number of min-plus products\footnote{The min-plus product (or $(\min,+)$-product) of matrices $A$ and $B$ is the matrix $C = A \oplus B$ with entries $C_{i,j} = \min_k\{A_{i,k} + B_{k,j}\}$.} of row-monotone or column-monotone bounded-difference matrices\footnote{In \cite{BGSV19}, the authors did not explicitly mention that the matrices were row-monotone or column-monotone. However, when we view Dyck edit distance as a maximization problem instead of a minimization problem, both problems can be cast as bounded monotone min-plus products.}. Since Dyck Edit Distance is a special instance of Language Edit Distance, it follows from \cite{BGSV19} that both Dyck Edit Distance and RNA Folding can be solved in the same time complexity as the bounded monotone min-plus product.

Recent work has shown that the online bounded monotone min-plus product (where columns of the second matrix are revealed one at a time, as in $\mathsf{OMv}$) is computationally equivalent to $\mathsf{OMv}$ \cite{HP25}. Consequently, extending the techniques of \cite{DG24} to online Scored Parsing yields an online algorithm for Dyck Edit Distance and RNA Folding with $\mathsf{OMv}$ running time.

\subsection{Static Scored Parsing and the Min-Plus Product}
First, we will give a brief overview of how to solve Scored Parsing using the (min,+)-product, as in \cite{BGSV19}, since both RNA Folding and Dyck Edit Distance can be reduced to Scored Parsing. Let $G = (N,T,P,S)$ be a CFG with score function $s : P \mapsto \mathbb{Z}_{\geq 0}$. Let $\mathcal{F}_N$ be the set of all functions mapping nonterminals to $\mathbb{Z}_{\geq 0} \cup \{\infty\}$. We re-define the ``plus'' operation $+ : \mathcal{F}_N \times \mathcal{F}_N \mapsto \mathcal{F}_N$ as follows:
\begin{align*}
    (F_1+F_2)(X) &\coloneqq \min \{s(X \rightarrow YZ) + F_1(Y) + F_2(Z) \mid (X \rightarrow YZ) \in P\}\text{ for }F_1, F_2 \in \mathcal{F}_N, X\in N.
\end{align*}
For a subset $\mathcal{F}$ of $\mathcal{F}_N$, we define the operation $\min \mathcal{F}$ as returning a function that takes the point-wise minimum:
\begin{align*}
    (\min \mathcal{F})(X) &\coloneqq \min \{F_i(X) \mid F_i \in \mathcal{F}\} \text{ for all }X\in N.
\end{align*}
Additionally, we define a special ``null'' function $E \in \mathcal{F}_N$ such that $F + E = E + F = F$ for all $F \in \mathcal{F}_N$ and $\min (\mathcal{F} \cup \{E\}) = \min (\{E\} \cup \mathcal{F}) = \min \mathcal{F}$ for all $\mathcal{F} \subseteq \mathcal{F}_N$.

For two matrices $A,B$ with entries in $\mathcal{F}_N$, we define their min-plus product $A \oplus B$ to be 
\begin{align*}
    (A \oplus B)_{i,j} = \min_k \{A_{i,k} + B_{k,j}\}
\end{align*}
using the above definitions for   ``$\min$'' and ``$+$''. Note that our $\oplus$ operation is not the usual min-plus product between two integer-valued matrices. However, for a constant-size grammar, our redefined min-plus product on functions can be reduced to a constant number of the original min-plus products \cite{BGSV19}. Thus, with a slight abuse of notation, we ignore the calculations under the hood.

Given a string $S \in T^*$ of length $n$, we define an $n+1$ by $n+1$ matrix $A$ whose entries are in $\mathcal{F}_{N}$. We set $A_{i,i+1}(X) = \min \{s(X \rightarrow S[i]) \mid (X \rightarrow S[i]) \in P \} \cup \{\infty\}$ for $i \in \fragment{1}{n}$; $A_{i,i} = E$ for $i \in \fragment{1}{n+1}$; and we set all other entries to be the function that maps every nonterminal to $\infty$.

Then we calculate the transitive closure \( A^+ \) of \( A \), defined as
\[
A^+ = \min \{A^{(i)} \mid i \in \fragment{1}{n+1}\},
    `\]
where \( \min \) is taken component-wise, and
\[
A^{(1)} = A, \quad
A^{(i)} = \min \{A^{(j)} \oplus A^{(i-j)} \mid j \in \fragment{1}{i-1}\}.
\]
Then the function $A^+_{i,j}$ evaluated at $X$ gives us the minimum score to generate $S\fragmentco{i}{j}$ from $X$, which we denote by $\minscore X \xRightarrow{\ast} S\fragmentco{i}{j}$. Hence the answer to Scored Parsing can be read off from $A^+_{1,n+1}(\texttt{S})$, where \texttt{S} conventionally denotes the start symbol. \cite{BGSV19} show that $A^+$ can be calculated in time $\Tminplus$ using a natural adaptation of Valiant's Parser.

\subsection{Online Scored Parsing and Online Min-Plus Product}
In Online Scored Parsing, we are given a CFG $G=(N,T,P,S)$ with score function $s : P \mapsto \mathbb{Z}_{\geq0}$. The string $S \in T^*$ is revealed to us one character at a time. At time step $t$, we need to report $\minscore \texttt{S} \xRightarrow{\ast} S\fragment{1}{t}$ before receiving the next character.

As a short-hand notation, for a string of nonterminals $S$ and two indices $i,j$, we define a function $F_{\fragment{i}{j}} \in \mathcal{F}_N$ such that for every nonterminal $X \in N$, $F_{\fragment{i}{j}}(X) = \minscore X \xRightarrow{\ast} S\fragment{i}{j}$. We define $F_{\fragmentco{i}{j}}$, $F_{\fragmentoc{i}{j}}$, and $F_{\fragmentoo{i}{j}}$ correspondingly. Lastly, we define the \textsf{Min-Plus-OMv} Problem, where we replace the boolean product in \OMv with the min-plus product.
\begin{lemma}\label{lem:online_queries}
    Let $\mI = \fragmentco{p}{p+s}$ be an interval of positions from $S$. Consider at most $s$ queries $Q_{p+s}, \dots, Q_{p+2s-1}$, where in query $Q_t$, we are given the set $\{F_{\fragment{i}{t}} \mid  i \in \fragment{(p+s)}{t}\}$, and we need to compute the set $A_t = \{F_{\fragment{i}{t}} \mid i \in \mI\}$. There is a randomized algorithm answering all queries in total time $\Oh(\TminplusOMv(s))$ and succeeds with high probability.
\end{lemma}
Before proving \cref{lem:online_queries}, let us first see how we can use \cref{lem:online_queries} to solve Online Scored Parsing. See \cref{alg} for a pseudo-code.

\begin{algorithm}[th]
    \KwInput{Context-free grammar $G = (N,T,P,S)$, score function $s : P \mapsto \mathbb{Z}_{\geq 0}$, and string $S \in T^*$, where character $S[t]$ is revealed at time step $t$.}
    \KwOutput{$\minscore \texttt{S} \xRightarrow{\ast} S\fragment{1}{t}$ for every time step $t$.}
    
    Let $\mathcal{I} = \{\}$;
    
   \For{$t = 1, 2, 3, \dots$} {
   $Q_t = \{F_{\fragment{t}{t}} \}$;
   
   \For{$\ell = |\mathcal{I}|$ to $1$}{$Q_t = Q_t \cup \mathcal{I}[\ell].\text{query}(Q_t)$;}

   Update $\mathcal{I}$ to be a partition of $S\fragment{1}{t}$ into intervals whose lengths correspond to the binary representation of $t$, with the largest interval at the beginning and the smallest interval at the end;\label{alg:line}

   \textbf{return} $F_{\fragment{1}{t}}(\texttt{S})$.
   }

    \caption{Solving Online Scored Parsing using \cref{lem:online_queries}} 
    \label{alg}
\end{algorithm}

\begin{lemma}
    The Online Scored Parsing Problem on a length-$n$ string can be solved in time $\TminplusOMv(n)$.
\end{lemma}
\begin{proof}
    In \cref{alg:line} of \cref{alg}, we maintain a partition of the part of $S$ that we have seen so far according to the binary representation of $|S|$, and intervals get created or removed as $|S|$ increases to $n$. Upon the creation of every interval, we initialize the algorithm in \cref{lem:online_queries} on it. Note that an interval of size $s$ will remain unchanged for $s$ queries, as the bit represented by that interval gets reset every $s$ increments. Then the correctness of the algorithm is immediate: At time step $t$, the intervals answer query $Q_t$ from right to left, extending $F_{\fragment{t}{t}}$ all the way to $F_{\fragment{1}{t}}$. 

    To see the desired running time, note that an interval of size $s$ is created $O(n/s)$ times. So the total running time is $\sum_{s=0}^{\log(n)}\frac{n}{2^s}\cdot \TminplusOMv(2^s)$. Suppose $\TminplusOMv(s) = \frac{s^c}{2^{\Omega(\sqrt{\log(s)})}}$ for some $2 \leq c \leq 4$, then we have 
    \begin{align}
        \sum_{s=0}^{\log(n)}\frac{n}{2^s}\cdot\frac{2^{cs}}{2^{\Omega(\sqrt{\log(s)})}} = n\cdot \sum_{s=0}^{\log(n)}2^{(c-1)s-\Omega(\sqrt{\log(s)})} = \frac{n^c}{2^{\Omega(\sqrt{\log(n)})}}.\label{eqn:online_RNA_helper_timebound}
    \end{align}
    
\end{proof}

Now we are ready to prove \cref{lem:online_queries}.
\begin{proof}
    In the preprocessing stage, we prepare the $s+1$ by $s+1$ matrix $A^+$ for $S\fragmentco{p}{p+s}$, which can be done in time $\Tminplus(s)$ \cite{BGSV19}. We index the rows and columns of $A^+$ with $\fragment{p}{p+s}$. Additionally, we maintain a matrix $M$ that grows column by column after every query. We index the rows of $M$ with $\fragment{p}{p+s}$ and the columns of $M$ with $\fragment{p+s}{t}$, where $t$ is the current time step. Each entry $H_{i,j}$ contains $F_{\fragmentco{i}{j}}$. Initially, $M$ consists of only one column, and we set $M_{i,p+s} = A^+_{i,p+s}$ for all $i \in \fragment{p}{p+s}$.

    In query $t$, we receive the set $\{F_{\fragment{i}{t}} \mid i \in \fragment{(p+s)}{t}\}$, and we arrange these functions into a vector $v_t$ indexed $\fragment{p+s}{t}$, where $v_{t_{i}} = F_{\fragment{i}{t}}$. Then the first $s$ entries of $u_t = A^+ \oplus (M \oplus v_t)$ hold the answers $\{F_{\fragment{i}{t}} \mid i \in \mI = \fragmentco{p}{p+s}\}$. Note that in the next query, $v_{t+1}$ is one element longer than $v_t$, and note that $u_t$ is exactly what column $t+1$ of $M$ asks for, so we append $u_t$ as a new column to $M$ to prepare for the next query.

    Now let us analyze the time complexity of all $n$ queries. Note that we are always min-plus multiplying $A^+$ with a length $s+1$ vector. Thus, we can initialize a Min-Plus-OMv instance for $A^+$, which runs in total time $\Oh(\TminplusOMv(s))$. However, the min-plus products of $M$ and $v_t$ are trickier, as both are changing in size. To tackle this, we can maintain a partition of $M$ in accordance with the binary representation of the current number of columns in $M$, as in \cite{DG24}. $M$ is then partitioned into a number of smaller square matrices, for each of which we can initialize a \textsf{Min-Plus-OMv} instance. At the end, we obtain $M \oplus v_t$ by combining the results from different \textsf{Min-Plus-OMv} instances. Since a $2^k$ by $2^k$ matrix gets created $\frac{s}{2^k}$ times, and there are $\frac{s}{2^k}$ of them upon a single creation, the total running time for computing $M \oplus v_t$ throughout all $s$ queries is given by
    \begin{align}
        \sum_{k=0}^{\log(s)}\frac{s}{2^k}\cdot \frac{s}{2^k} \cdot \TminplusOMv(2^k) = \Oh(\TminplusOMv(s))\label{eqn:online_RNA_total_timebound}
    \end{align}
using a similar analysis as \cref{eqn:online_RNA_helper_timebound}.
\end{proof}

\onlineRNA
\begin{proof}
   Online RNA Folding and Online Dyck Edit Distance can be cast as instances of Online Scored Parsing on bounded monotone grammars. Thus, all of the matrices and vectors involved in \cref{lem:online_queries} are bounded and monotone. Using the $\Oh(n^3/2^{\Omega(\sqrt{\log(n)})})$-time algorithm for bounded monotone \textsf{Min-Plus-OMv} \cite{OMVwilliams, HP25}, we obtain the desired running time. 
\end{proof}